\documentclass[draftclsnofoot,12pt,onecolumn] {IEEEtran}

\usepackage{mdwlist}

\usepackage{flushend}
\usepackage{url}
\usepackage{graphicx}
\usepackage{subfigure}
\usepackage{xcolor,colortbl}
\usepackage[mathscr]{eucal}
\usepackage{amsthm }
\usepackage{amssymb}
\usepackage{cases}
\newtheorem{heuristic}{Heuristic}

\usepackage{mathptmx}

\usepackage{wrapfig}

\usepackage{subfigure}
\usepackage{algorithm}
\usepackage[noend]{algorithmic}

\floatname{algorithm}{Algorithm}

\usepackage{mathptmx}
\usepackage{mdwlist}
\usepackage{flushend}
\usepackage[mathscr]{eucal}
\usepackage{xcolor,colortbl}
\usepackage{url}
\usepackage[mathscr]{eucal}
\usepackage{amsthm }
\usepackage{amssymb}
\usepackage{mdwlist}
\usepackage{url}
\usepackage{graphicx}
\usepackage{subfigure}

\newtheorem{definition}{Definition}[section]
\newtheorem{lemma}{Lemma}[section]
\newtheorem{theorem}{Theorem}[section]

\newtheorem{corollary}{Corollary}[section] 

\usepackage{mathptmx}
\usepackage{flushend}
\usepackage{xcolor,colortbl}

\definecolor{mygray11}{gray}{.99}
\definecolor{mygray10}{gray}{.89}
\definecolor{mygray9}{gray}{.79}
\definecolor{mygray8}{gray}{.69}
\definecolor{mygray7}{gray}{.59}
\definecolor{mygray6}{gray}{.49}
\definecolor{mygray5}{gray}{.39}
\definecolor{mygray4}{gray}{.29}
\definecolor{mygray3}{gray}{.19}
\definecolor{mygray2}{gray}{.1}
\definecolor{mygray}{gray}{.01}

\definecolor{mygray7test}{gray}{.973}
\definecolor{mygray6test}{gray}{.906}
\definecolor{mygray5test}{gray}{.7532}
\definecolor{mygray4test}{gray}{.5998}
\definecolor{mygray3test}{gray}{.42}
\definecolor{mygray2test}{gray}{.25}
\definecolor{mygraytest}{gray}{.01}

\ifCLASSOPTIONcompsoc
\else
\fi
\ifCLASSINFOpdf
\else
\fi
\begin{document}

%
\title{\huge  An Improved Solution for Restricted and Uncertain TRQ   }
%
%
%
%

\author{Jack~Wang 
\IEEEcompsocitemizethanks{\IEEEcompsocthanksitem Contact information: cszjwang@gmail.com. 
\protect\\

}
\thanks{}}

\IEEEcompsoctitleabstractindextext{%
\begin{abstract}
 CSPTRQ is an interesting problem and its has attracted much attention.  The CSPTRQ is a variant of the traditional  PTRQ.  As objects moving in a constrained-space are common, clearly, it can also find many applications. At the first sight, our problem  can be easily tackled by extending existing methods used to answer the PTRQ. Unfortunately, 
 those classical techniques are not well suitable for our problem, due to a set of new challenges.  
 We develop  targeted solutions and  demonstrate the efficiency and effectiveness of the proposed methods through extensive experiments.    
 \end{abstract}

}

\maketitle

\IEEEdisplaynotcompsoctitleabstractindextext

%
\IEEEpeerreviewmaketitle





 \section{Introduction}\label{sec:1}
 \vspace{-1ex}
 
 The range query  as one of fundamental operations in moving object search systems   has attracted a lot of attention  in the past decades  \cite{RuiZhang:Optimized,bugragedik:processing,KunLungWu:Incremental,HaojunWang:Processing,YufeiTao:The,HodaMokhtar:On,SunilPrabhakar:Query,MohamedFMokbel:SINA,MohamedFMokbel:SOLE,HaiboHu:AGeneric,MuhammadAamirCheema:multi,BinCui:IMPACT,DariusSidlauskas:Parallel}.   
 A database server usually only stores the discrete location information due to various reasons such as  the limited network bandwidth and  battery power of the mobile devices \cite{dieterpfoser:capturing,reynoldcheng:querying}.
 This fact implies that the current specific position of a  moving object $o$ is uncertain before  obtaining  the next (sampled) location information, which can lead to the incorrect answer if we simply take the recorded location  (stored in the database) as the current position of $o$. 
 In order to tackle the aforementioned problem, the idea of incorporating \textit{uncertainty} into the moving object data has been proposed \cite{OuriWolfson:Updating}.  
 A widely-used uncertainty model  is to use a closed region (known as  \textit{uncertainty region}) together with a probability
 density function (PDF), which is used to denote the object's \textit{location distribution}\cite{OuriWolfson:Updating,reynoldcheng:querying}.  
 
 From then on, probabilistic range query (PRQ) as a derivative  version of the traditional range query was naturally presented, and  many outstanding works   addressed this problem (see e.g., \cite{brucesechung:processing,aprasadsistal:Modeling,dieterpfoser:capturing,gocetrajcevski:probabilistic,jinchuanchen:efficient,aprasadsistal:querying,gocetrajcevski:managing,reynoldcheng:querying,meihuizhang:effectively}). 
 In existing results, one of important branches  is to address the PRQ over  objects   moving  freely (without predefined routes)  in  two-dimensional (2D) space (see e.g., \cite{jinchuanchen:efficient,meihuizhang:effectively,reynoldcheng:querying}). Our work generally falls in the aforementioned branch. 

 \noindent \textbf{{Motivations}.}
 A common fact is that  users usually are  interested in the objects  being located in the query range $R$ with higher probabilities.  Several classical papers  (see e.g., \cite{jinchuanchen:efficient,yufeitao:indexing,yingzhang:Efficient}) already considered this fact and studied the probabilistic threshold range query (PTRQ).     Existing results are mainly developed for the case of non-constrained 2D space (i.e., no obstacles exist). To our knowledge, the constrained-space probabilistic threshold range query (CSPTRQ) has not been studied yet.   Moreover, we realize that more and more intelligent terminals have been configured with touch screens  by which one can input the query requirement using the finger or  interactive pen \cite{par:highprecision,ArneLDuwaer:dataprocessing}. An obvious fact is that a more generic  shaped query range should be better for the user experience, and can also improve the flexibility of a system itself. Existing works (see e.g., \cite{aprasadsistal:Modeling,aprasadsistal:querying}) already adopted the general polygon as the query range. Those results are mainly developed  from the theoretical perspective. 
 Specifically, this work studies the CSPTRQ supporting a  {generic shaped} query range, for  moving  objects. 
 
 The CSPTRQ can be used in a lot of applications, as  objects moving in a constrained 2D space are common in the real world.  For example,   mobile robots are  already used to rescue survivors after a disaster such as an earthquake \cite{Robin:DisasterRobotics}. The location information of  robots is   collected and  stored on the database server.  A typical application for dispatching scattered robots to a specific location is   retrieving the identities of the robots that are currently located in a given region with no less than a predefined (e.g., $75\%$)  probability; here robots usually move freely  but can be  blocked by various obstacles (e.g., rocks, buildings). As another example, in  the information warfare the  location information of  combat machineries  is collected and usually stored on  the military database \cite{OuriWolfson:moving,Martin:Mobile}. A typical application for the coordination combat  is  retrieving the identities of the friendly machineries (e.g.,  tanks and panzers) that are currently located in a given region with no less than a specific (e.g., $85\%$)  probability; here  objects such as tanks and panzers   usually move freely without predefined routes but can be blocked by various obstacles (e.g., lakes, hills).

 \noindent \textbf{{Challenges}.}
 At the first glance, the CSPTRQ can be easily tackled by directly extending  existing methods used to answer the PTRQ. As a matter of fact,  there are several new challenges. (\romannumeral 1)  The CSPTRQ needs to handle a  set of obstacles, and so the workload  is larger, implying that to achieve a quick response time is more challenging.  (\romannumeral 2) With the presence of obstacles, the uncertainty region  $u$  is usually  a   complicated geometry (see Section \ref{sec:preliminaries} for  more details), rendering that   the subsequent computation is more difficult. (\romannumeral 3) In   a non-constrained space,  
 $u$ can be easily obtained (almost) without taking the precomputation cost, and thus existing methods  usually   pre-compute a set of bounds based on the uncertainty region $u$ and the probability density function (PDF). These bounds are used to prune/validate unqualified/qualified objects, and can significantly improve the performance, especially when  they are correctly indexed using the R-tree like data structure.  In the context of our concern, the precomputation time is rather long (up to the \textit{hour} level) \textit{even if} we only pre-compute the uncertainty
 regions. (See Section \ref{sec:preliminaries} for  more detailed discussion about \textit{bounds} and the \textit{precomputation}.)  Imagine if we further pre-compute lots of bounds, the overall precomputation
 time should be larger. With these challenges (particularly, the third one) in mind, we have to resort to other proposals. 
 
 Another method used to answer the constrained-space probabilistic range query (CSPRQ) \cite{ZhijieWang:prqumo} can be easily extended to tackle our problem. Unfortunately, a simple adaptation of this method is inefficient, due to its weak pruning/validating capability. (See Section \ref{sec:preliminaries} for  more details about the baseline method.) 
 
 Overall, we are confronted with the following troubles: (\romannumeral 1) those classical techniques (used to answer the PTRQ)  have powerful pruning/validating capabilities,  but are not well suitable for the context of our concern, and (\romannumeral 2) the method used to answer the CSPRQ is easily incorporated, but to find a feasible and powerful pruning/validating mechanism is not easy.
 
 \noindent \textbf{{Contributions}.}
 A casual trifle, shopping in a supermarket, gives  us the initial inspiration.  The shopper  freely chooses his/her wanted commodities  and finally obtains them by paying the bill. Clearly, it is a \textit{swap}: money $\longleftrightarrow$ commodities.
 This trifle reminds us that \textit{swapping}  can usually obtain the wanted things.
 With this (concept)  in mind, we revisit our problem and develop  our first  idea --- \underline{s}wapping the order of geometric operations, which  simplifies   the computation and can prune/validate some objects without the need of computing their uncertainty regions. 
 After this, by carefully considering the details, we  realize that the result obtained in the previous step possibly is a fake result, which stems from the \textit{location unreachability}. The natural method to eliminate the fault is  inefficient. Instead, our strategy is to  take advantage of the  location  unreachability. This method not only eliminates the possible fault, but also prunes some objects in the early stages.  
 All strategies developed above actually belong to \textit{spatial}  pruning/validating mechanisms.  
   
 When we strive to seek the  \textit{threshold} pruning/validating mechanisms, suddenly, we realize an interesting fact --- the CSPTRQ can be classified into two forms:  explicit  and  implicit ones (they can have different   solutions, performance results, and purposes/applications). The former returns a set of tuples in form of ($o$, $p$) such that $p\geq p_t$, where $p$ is the  probability of the moving object $o$ being located  in the query range $R$, and $0\leq p_t\leq 1$ is a given  probabilistic threshold. A potential purpose/application  is like: listing the objects (e.g., tanks) that are currently located in the region $R$ with no less than the $80\%$  probability in the descending order according to their appearance probabilities; it is similar to the following: listing the universities that are with no less than 80 points in the descending order according to their points, where the points usually be evaluated using a variety of indicators such as the publications in Nature/Science. (Remark: the traditional \textit{probabilistic range query} (PRQ) usually refers to the explicit form but $p_t=0$. Thus, the immediate purposes/applications of the explicit CSPTRQ are the similar as the ones of the traditional PRQ.)  
 In contrast, the latter  returns a set of objects, which have  probabilities higher than  $p_t$ to be located in   $R$. A potential purpose/application  is like: returning the number of objects (e.g., mobile robots) that are currently located in the region $R$ with no less than the $75\%$  probability. (Remark: the traditional \textit{probabilistic threshold range query} (PTRQ)  usually refers to the implicit form. Thus, the immediate purposes/applications of the implicit CSPTRQ are the similar as the ones of the traditional PTRQ.) 
 See  Figure \ref{fig:2a} for example. We assume there is no obstacles, $R$ is a rectangle, and the location of $o$ follows uniform distribution in $u$ for simplicity. Suppose $p_t=0.2$, the answer of explicit query is $\{ (o_2,50\%),$ $(o_3,50\%),$ $(o_4,25\%)\}$, while the answer of implicit query is $\{o_2,o_3,o_4\}$.

 \begin{figure}[h]
   \centering
      \includegraphics[scale=.45]{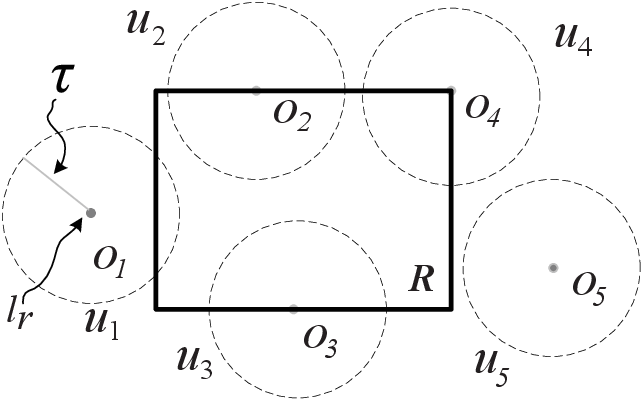} 
             \vspace{-2ex} 
  \caption{\small  Example of explicit and implicit queries, where $l_r$ denotes the recorded location, $\tau$ denotes the distance threshold, and $u_i$ denotes the uncertainty region of object $o_i$ ($i\in[1,2,\cdots,5]$)} 
 \label{fig:2a}
 \end{figure}


 
 The second main idea  is inspired by the  \textit{evolutionary algorithms} \cite{Kenneth:Evolutionarycomputation}.    
 A typical characteristic of  evolutionary algorithms is the repeated application of a set of predefined operators; and each iteration  can be generally looked as a refinement of the previous result.   
 This reminds us to compute the appearance probability $p$ in  a \underline{m}ulti-step manner, and thus objects that are obviously unqualified can be pruned in the early steps. This idea is especially effective  when the locations of  objects do not follow  uniform distribution in their uncertainty regions. The multi-step strategy yields a set of \textit{threshold} pruning/validating rules, which are employed by the explicit query. As the implicit query does not need to return the appearance probabilities of qualified objects,  an \underline{e}nhanced multi-step   strategy is naturally developed, which includes an \textit{adaptive} pruning/validating mechanism  and a two-way test mechanism. Furthermore, we further optimize our solutions based on a new insight --- different \textit{candidate moving objects} may share the same \textit{candidate restricted areas}.  
 In summary,   our  contributions are as follows: 
 \begin{itemize*}
  \item We propose the CSPTRQ, and show that (\romannumeral 1) it can be used in many applications;  (\romannumeral 2) the classical methods used to answer the traditional PTRQ are not well suitable for the context of our concern; and (\romannumeral 3) a simple adaptation of the method used to answer the CSPRQ is inefficient.
 \item We realize the CSPTRQ can be classified into two forms: explicit and implicit ones. We formally formulate them, and offer insights into  their properties.
 \item We develop techniques  to answer the explicit query, and then  extend them to answer the implicit query. Our solutions are  simple but without loss of efficiency. 
 \item We give the detailed theoretical analysis for our algorithms. While we  focus on the CSPTRQ in this paper, (part of) our techniques  can  be immediately extended  to other types of probabilistic threshold queries.  
 \item We  experimentally evaluate our algorithms using both real and synthetic data sets. The experimental results demonstrate the efficiency and effectiveness  of the proposed algorithms. From the experimental results, we can further perceive  the difference between explicit and implicit queries. This interesting finding is valuable especially for the topics of other types of probabilistic threshold queries. 
 \end{itemize*}

 \noindent \textbf{{Paper}  organization.}
 We review the related work in Section \ref{sec:2}. We formally formulate our problem and present  a baseline method in Section \ref{sec:3}.  The proposed methods for answering the explicit and implicit CSPTRQs are addressed in Section \ref{sec:eprq} and \ref{sec:iprq}, respectively. We further optimize our solution based on a new insight in Section \ref{sec:further optimziation}. We evaluate the performance of our proposed methods  through extensive experiments in Section \ref{sec:6}. 
 Finally, we conclude  this paper with several interesting  research topics  in Section \ref{sec:7}.

 \section{Related work} \label{sec:2}
 

 \noindent \textbf{{Range}  query over moving objects. } 
 Most of the representative works on \textit{range query over moving objects} have been mentioned in Section \ref{sec:1}. 
 A common aspect of those works is  not to capture the location uncertainty. In other words, they assume the current location of any object $o$ is equal to the recorded location (stored on the database server). In contrast,  we assume the current location of $o$ is uncertain.  

 \noindent \textbf{{Uncertainty} {m}odels. }
 We also mentioned many outstanding works on \textit{PRQ over uncertain moving objects} in Section \ref{sec:1}.  One of important branches assumed that  objects  move  freely (without predefined routes)  in  2D space. In this branch, there are several typical \textit{uncertainty} models like, the   free moving uncertainty (FMU) model \cite{reynoldcheng:querying,OuriWolfson:Updating}, the   moving object spatial temporal (MOST) model \cite{aprasadsistal:Modeling}, the  uncertain moving object (UMO) model \cite{meihuizhang:effectively}, the 3D cylindrical (3DC) model  \cite{gocetrajcevski:managing,dieterpfoser:capturing},  and the necklace uncertainty (NU)   model \cite{gocetrajcevski:Uncertain,BartKuijpers:Trajectory}. 
 Another important branch assumed  that  objects move on  predefined routes \cite{reynoldcheng:querying} or road networks \cite{KaiZheng:probabilistic}. They usually  adopt the line segment uncertainty  (LSU) model \cite{brucesechung:processing,reynoldcheng:querying}  to capture the location uncertainty. 
 These models   have different assumptions and purposes (e.g., 3DC  and NU models are suitable for querying the trajectories of moving objects),  but have  their own advantages (note: it is a difficult task to say which one is the
 best. Please refer to \cite{ZhijieWang:prqumo} as a summary on the differences of these models and their assumptions). 
 The  model used in  \cite{ZhijieWang:prqumo}  roughly follows the FMU model, but it is different from the FUM model, as it  introduces the concept of restricted areas (i.e., obstacles). Here we  dub it  the extensive free moving uncertainty (EFMU) model for clearness.   
 
 Though our work also uses the EFMU model, there are at least two differences: (\romannumeral 1) our work investigates  CSPTRQs (including  explicit  and  implicit ones) rather than the CSPRQ, and (\romannumeral 2) our work employs  a more generic shaped query range.
 
 
 

 
 \noindent \textbf{{Probabilistic} {t}hreshold {r}ange {q}uery. }
 According to the theme of this paper, we classify \textit{PTRQ}s into two subcategories: PTRQs for  moving objects  and the ones  for other uncertain data (note: the terms  ``PRQ'' and ``PTRQ'' are somewhat abused in the  literature, we take those papers, which explicitly discussed the probabilistic threshold,  as the  related work of  the PTRQ). 
 
 Many excellent works addressed the PTRQ for moving objects. For example, Chung et al. \cite{brucesechung:processing} addressed the PTRQ for objects moving in one-dimensional (1D) space. In contrast, we focus on the objects moving in 2D space. Zhang et al. \cite{meihuizhang:effectively} studied the PTRQ over objects moving in  2D space. They proposed the UMO model, in which they assumed both the \textit{distribution} of velocity and the one of location    are available at the update time.  In contrast,  we do not need  to know the  velocity (as well as its distribution), instead we assume the specific location   of any object $o$  is available at the update time. Moreover,  the used model in this paper is the EFMU model, which  considers the existence of restricted areas. Zheng et al. \cite{KaiZheng:probabilistic} studied the PTRQ for  objects moving on the road networks. They proposed the UTH model that is developed for querying the trajectories of moving objects. In contrast, this paper is not interested in querying the trajectories, and it focuses on the objects moving in the constrained 2D space where no predefined route is given.

 
 There are  many classical papers that studied the PTRQ for other uncertain data. For example, Cheng et al. \cite{reynoldcheng:efficient} addressed the PTRQ over 1D uncertain data (e.g. sensor data), they presented a clever idea, using a tighter  bound (compared to the  MBR of the uncertainty interval),  called \textit{x-bound},    to   reduce the search cost. Later, Tao et al. \cite{yufeitao:range} extended this idea to   multi-dimensional uncertain data. They proposed a classical technique,  probabilistic constrained region (PCR), which consists of a set of precomputed bounds, called  \textit{p-bound}s. This classical  technique is not well suitable for the context of our concern, Section \ref{sec:1} has shown the reasons (more detailed discussion will be given in Section \ref{sec:preliminaries}).     Chen et al. \cite{jinchuanchen:efficient} studied the PTRQ for such a scenario where the location of query issuer is uncertain (a.k.a, location based PTRQ);  several smart ideas such as the \textit{query expansion} were developed. They assumed the query range $R$ and uncertainty region $u$ are rectangles, and  focused on the non-constrained space, and thus employed the \textit{p-bound}s technique. In contrast,  both $R$ and $u$ used in our work are more complex, and we focus on the constrained space, where the \textit{p-bound}s technique has some limitations (again,  Section \ref{sec:1} has shown the reasons). Moreover, our work does not belong to the location based PTRQ.   
 
 \noindent \textbf{{Other} probabilistic threshold  queries.}
 There are also many representative works that addressed other probabilistic threshold  queries (PTQs); those works are clearly different from ours. For instance, Zhang et al. studied the \textit{location based} {probabilistic threshold range \textit{aggregated}  query} \cite{yingzhang:Efficient}. Hua et al. \cite{Minghua:rankingQueries} addressed the  probabilistic threshold \textit{ranking} query on uncertain data. The probabilistic threshold \textit{KNN} query over uncertain data was investigated by Cheng et al. \cite{reynoldcheng:evaluatingPTKNNQ}. Yuan et al. \cite{yeyuan:efficientlyAPTBSPQ}  discussed the probabilistic threshold \textit{shortest path} query over uncertain graphs. The \textit{general} PTQ for arbitrary SQL queries that involve \textit{selections}, \textit{projections}, and \textit{joins} was studied by Qi et al. \cite{yinianqi:thresholdquery}.

 

 

 \section{Problem definition } \label{sec:3}

 \subsection{Problem settings and notations}\label{subsec:problem setting}
 Let $R$ be the query range. Let $r$ denote the restricted area, and   ${\mathscr R}$ be a set of     disjoint restricted areas. Let  $\mathbb{T}$ be a territory such that $\bigcup_{r\in \mathscr{R}}{r}\subset { \mathbb{T}}$.  
 Let $o$ denote the moving object, and  ${\mathscr O}$ be a set of  moving objects. 
 Let $l_r$ be the latest recorded location (stored on the database server) of $o$, and $l_t$ be the location of $o$ at an arbitrary instant of time $t$. We assume that  $l_t\notin\bigcup_{r\in {\mathscr R}}{r}$ and $l_t\in {\mathbb{T}}-\bigcup_{r\in {\mathscr R}}{r}$. 
 Let $\tau$ be the distance threshold  of $o$. We assume any object $o$  reports \textit{its new location} to the server once  $dist(l_{t_n},l_r)\geq$ $\tau$, where $l_{t_n}$ denotes its current specific location, $dist(\cdot)$ denotes the Euclidean distance. 
 Finally, for any two different objects $o$ and $o^\prime$, we assume they cannot be located in the same location at the same instant of time $t$, i.e., $l_t\neq l_t^{\prime}$.

 
 
 
 
 We  model both the query range and  restricted areas  as  the arbitrary shaped polygons{\small \footnote{\small  Any curve can be approximated into a polyline (e.g., by an interpolation method). Hence in theory any shaped restricted area  or query range  can be approximated into a polygon.}}. We capture the location uncertainty using  two components  \cite{reynoldcheng:querying,OuriWolfson:Updating}.


 
 \begin{definition}[\small Uncertainty region]
 {
 {The {uncertainty region} of a moving object $o$ at a given time $t$, denoted by $u^t$, is a closed region where $o$ can always be found.}
 
 }
 \end{definition}

 \begin{definition}[\small Uncertainty probability density function]
 {
 {The {uncertainty probability density function} of $o$ at time $t$, denoted by $f^t(x,y)$, is a probability density function (PDF) of $o$'s location at a given time $t$; its value is 0  if $l_{t}$ $\notin$ ${ u}^{t}$.}
 }
 \end{definition}

 The PDF has the property that $\int_{{u}^t} f^t(x,y)dxdy=1$. In addition,   under the distance based  update policy (a.k.a., dead-reckoning policy \cite{OuriWolfson:Updating,reynoldcheng:querying}), for any two different  time $t_1$ and $t_2$ ($t_1,t_2\in$ ($t_r$,$t_n$]),   the following conditions always hold:  $u^{t_1}=u^{t_2}$ and $f^{t_1}(x,y)=f^{t_2}(x,y)$, 
 where  $t_r$ refers to the latest reporting time,   $t_n$ refers to the current time. 
 Hence,  unless stated otherwise, we  use $u$ and $f(x,y)$ to denote the uncertainty region and PDF of $o$, respectively.   (Remark:  if the time based update policy is assumed to be adopted, such a topic is more interesting and also more challenging, since the uncertainty region $u$ is to be a continuously changing geometry
 over time. See, e.g.,  \cite{ZhijieWang:prqumo} for a clue about the relation between the location update policy and the uncertainty region $u$.) 
 With the presence of restricted areas (i.e., obstacles), the uncertainty region $u$  under the distance based update policy can be formalized as follows. 
 \begin{equation}\label{equ:temp compute u}
 u=o.\odot-\bigcup_{r\in {\mathscr R}}r
 \end{equation}
 where $o.\odot$  denotes a circle   with the centre  $l_r$ and radius $\tau$. 
 We remark that, in the rest of this paper, we abuse the notation `$|\cdot|$', but its meaning should be clear from the context. 
 In addition, unless stated otherwise, a notation or symbol with a subscript `b'  usually refers to its corresponding minimum bounding rectangle (MBR). For instance, $R_b$ refers to the MBR of $R$. For convenience, Table \ref{tab:main symbols} summarizes the  notations used frequently in the rest of this paper.

 \begin{table}[h]
 \begin{center}
 \begin{tabular}{|p{.18\textwidth}| p{.6\textwidth}  | }\hline

 {{\small \textbf{Notations}}}&{{\small \textbf{Meanings}}} \\ \hline 
 



 {\small $\mathscr R^*$}  & {\small  the set of candidate restricted areas}	\\ \hline

  {\small $\mathscr O^*$}  & {\small the set of candidate moving objects} 	\\ \hline 
  
  {\small $R_b$}  & {\small  the minimum bounding rectangle of the query range $R$}  \\ \hline
 
 

 
 {\small $\tau$}  & {\small distance threshold } 	\\ \hline 
 
  {\small $l_r$}  & {\small  recorded location of a moving object $o$}	\\ \hline 
  
  {\small $o.\odot$}  & {\small  circle with  the centre $l_r$ and radius $\tau$}	\\ \hline 
  
   {\small $\mathscr I_r$}  & {\small  index of restricted areas } 	\\ \hline 
  
    {\small $\mathscr I_o$}  & {\small  index of moving objects }	\\ \hline 
  
 
 
 {\small $p_t$}  & {\small probabilistic threshold} 	\\ \hline 
 

 {\small $u_o$}  & {\small  outer ring of uncertainty region $u$} \\ \hline

  {\small $u_h^i$}  & {\small  the $i$th hole in uncertainty region $u$}	\\ \hline 
 
 {\small $\mathscr H$}  & {\small the set of holes in uncertainty region $u$}	\\ \hline 
 
  {\small $s$}  & {\small  intersection result between $R$ and $u$}	\\\hline 
 
  {\small $|s|$}  & {\small the number of subdivisions of $s$}	\\ \hline 
 
  {\small $s[i]$}  & {\small the $i$th subdivision of $s$}	\\ \hline 
 
  {\small $s[i]_o$}  & {\small  outer ring of $s[i]$}	\\ \hline 
 
 {\small $\mathscr H^*$}  & {\small the set of all holes in $s$}	\\  \hline 
 
 {\small $s_h^j$}  & {\small the $j$th hole among all the $|\mathscr H^*|$ holes of $s$}	\\  \hline



  
 
 {\small $\gamma$}  & {\small  reference value} 	\\

  \hline
 \end{tabular}
 \end{center}
 \vspace{-2ex}
 \caption{\small Notations and their descriptions}\label{tab:main symbols}
 \end{table}
 
 \vspace{-1ex}


 \subsection{Problem statement}\label{subsec: problem statement}
 Let $p_t$ be the probabilistic threshold, we have 
 \begin{definition}
 {Given a set $\mathscr R$ of restricted areas,  a set $\mathscr O$ of moving
 objects in a territory $\mathbb{T}$, and a query range $R$, 
 an explicit constrained-space probabilistic threshold range query (ECSPTRQ)   returns a 
 set of tuples in form of ($o$, $p$) such that   $p\geq p_t$, where $p$ is the  probability of $o$    being located in $R$, and is computed as
 \begin{equation}\label{equation:pj}
 p=\int_{u\cap R} f(x,y)dxdy
 \end{equation}
 
 }
 \end{definition}
 
 We note that $f(x,y)$=$\frac{1}{\alpha(u)}$ when the location of $o$  follows uniform distribution in its uncertainty region $u$,  where $\alpha(\cdot)$ denotes the area of this geometric entity. In this case,  we have

 \begin{equation}
 \label{equation:area divide area}
 p= \frac{\alpha(u\cap R)}{\alpha(u)}
 \end{equation}

 \begin{definition}
 {Given a set $\mathscr R$ of restricted areas,  a set $\mathscr O$ of moving
 objects in a territory $\mathbb{T}$, and a query range $R$, 
 an implicit  constrained-space probabilistic threshold range query (ICSPTRQ)   returns all the objects $o$ such that $p\geq p_t$,  where $p$ is the  probability of $o$  being located in  $R$, and is computed according to Equation (\ref{equation:pj}).
 }
 \end{definition}

 
 

 
 We remark that though the differences of two queries above are minor at the first glance,  we will present  different solutions respectively in Section \ref{sec:eprq} and \ref{sec:iprq}, and  show their different performance results in Section \ref{sec:6}.  Sometimes, we also use terms  the \textit{explicit query} and the \textit{implicit query} to denote the above two queries in the rest of this paper.  For ease of understanding the proposed methods,  we next introduce a baseline method.

 \subsection{Baseline method} \label{sec:preliminaries}

 The baseline method is a simple adaptation of the method in \cite{ZhijieWang:prqumo}. To save  space, we only present an overall framework of the baseline method. 
 
 \noindent \textbf{{Preprocessing} stage.}
 Here a twin-index is adopted (e.g., a pair of R-trees or its variant): one  is used to manage the set $\mathscr R$ of restricted areas; another is used to manage the set $\mathscr O$ of moving objects. To index restricted areas is  simple, since we model them as arbitrary polygons. Naturally,  we  can easily find the  MBR  of any restricted area $r$ ($\in \mathscr R$).  In order to manage the set $\mathscr O$ of moving objects, we here index them based on their  recorded locations $l_r$  and distance thresholds $\tau$. Specifically, for each object $o$, its MBR  is a square  centering at  $l_r$    with  $ 2\tau\times 2\tau$ size. For clearness, let   ${\mathscr I}_o$ and ${\mathscr I}_r$ be the  index of moving objects and the one of restricted areas, respectively. 
 
 \noindent \textbf{{Query} processing stage.}
 We first give two definitions  before discussing the details. 
 \begin{definition}[\small Candidate moving object]
 Given    a moving object $o$   and   the query range ${R}$,   $o$ is a candidate moving object  such that $R_b\cap o.\odot_b$$\neq$$\emptyset$.
 \end{definition}
 
 \begin{definition}[\small Candidate restricted area]
 Given    a moving object $o$  and a restricted area $r$,  $r$ is a candidate restricted area  such that $r_b\cap o.\odot_b$$\neq$$\emptyset$.
 \end{definition}
 
 Let ${\mathscr R^*}$ denote the set of candidate restricted areas, and ${\mathscr O^*}$ denote the set of candidate moving objects. There are several main steps for answering the implicit (or explicit) query. First, we search ${\mathscr O^*}$ on $\mathscr I_o$ using $R_b$  as the input (here most of unrelated objects are to be pruned). Second, for each object $o\in \mathscr O^*$, we search $\mathscr R^*$ on $\mathscr I_r$ using $o.\odot_b$ as the input (here most of  unrelated restricted areas are to be pruned).  We compute  $o$'s uncertainty region $u$, and then compute ``$u\cap R$''{\small \footnote{\small Note that, the  algorithm in  \cite{ZhijieWang:prqumo}  cannot support the generic shaped query range, and thus  some modifications   are necessary and inevitable when we compute $u\cap R$; moreover, the details of managing  complicated geometric regions (e.g., $u$) can be found in that paper. }}. After this, we compute  $p$ using Equation (\ref{equation:pj}). We put $o$ (or ($o,p$)) into the result if $p\geq p_t$.  Otherwise, we discard it and process the next object. After all candidate moving objects are handled, we finally return the result, in which all qualified objects are included.

 \noindent \textbf{{Update} stage.}
 When an object $o$ reports its new location to the server, we update the database record, i.e., $l_r$. At the same time, we update the index of moving objects, i.e., ${\mathscr I}_o$.
 
 \begin{figure}[t]
   \centering
   \subfigure[\small ]{\label{fig:1new:a}
      \includegraphics[scale=.43]{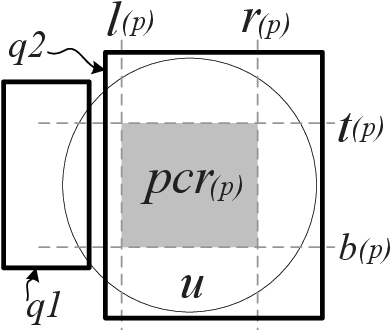}} 
   \subfigure[\small ]{\label{fig:1new:b}
       \includegraphics[scale=.43]{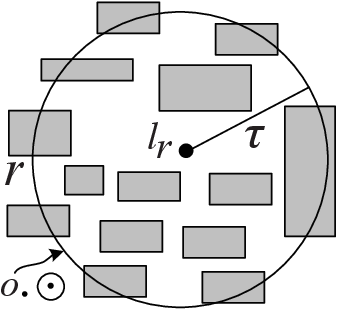}} 
   \subfigure[\small ]{\label{fig:1new:c}
       \includegraphics[scale=.43]{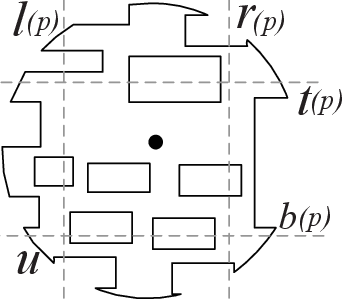}} 
   \subfigure[\small ]{\label{fig:1new:d}
       \includegraphics[scale=.32]{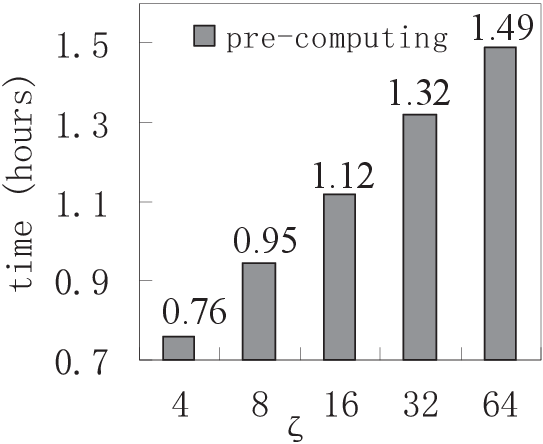}} 
             \vspace{-2ex} 
  \caption{\small Illustration of of p-bounds and the precomutaion. (a) The case of no restricted areas. (b) The case of existing restricted areas. (c) The uncertainty region. (d) The precomputaion time when $|\mathscr R|=|\mathscr O|=50k$. $\zeta$ denotes the number of edges in each restricted area $r$ (note: it may be somewhat difficult to understand this figure, and the readers can revisit it after reading Section \ref{sec:6}).  } 
  \label{fig:bounds and precomputation}
 \end{figure}
 
 \noindent \textbf{{{Discussion}}.}
 The readers may be curious why the baseline method does not employ existing \textit{threshold} pruning/validating mechanisms such as \textit{p-bounds} in \cite{yufeitao:range,jinchuanchen:efficient,yufeitao:indexing,yingzhang:Efficient}. Briefly speaking, a \textit{p-bound} of the uncertainty region $u$ (of the object $o$) is a function of $p$, where $p\in[0,0.5]$. A probabilistically constrained  region (PCR) with the parameter $p$, denoted by $o.pcr(p)$, consists of  four   \textit{p-bounds}, namely  $l(p)$, $r(p)$, $t(p)$ and $b(p)$,  see the four dashed lines in  Figure \ref{fig:1new:a}. The line $l(p)$ divides the uncertainty region $u$ (i.e., the circle) into two parts (on the left and right of $l(p)$ respectively), and the appearance probability of $o$ on the left part equals $p$. (Other three lines have similar meanings.) The grey region illustrates  $o.pcr(p)$. Assume the parameter $p$ in Figure \ref{fig:1new:a} is 0.2; moreover, assume the probabilistic threshold $p_t=0.8$, and if  $q1$  is the query range, then  $o$ is an unqualified object, and thus to be pruned.  In contrast, if $q2$  is the query range, then  $o$ is a qualified object, and thus to be validated. The example above illustrates the rationale of the classical \textit{p-bounds} technique. In a non-constrained space (i.e., no obstacles exist), all the uncertainty regions can be easily obtained (almost) without taking the precomputation cost, and thus  pre-computing a set of \textit{p-bounds} is feasible. However, in  the context of our concern,  the precomputation time is  rather long (up to the \textit{hour} level) \textit{even if} we only pre-compute the uncertainty regions. Figure \ref{fig:1new:d} reports the time of \textit{pre-computing} a set of uncertainty regions.   Imagine if we further pre-compute lots of \textit{p-bounds}, then the overall precomputation time should be larger. This is the  main reason why  the \textit{p-bounds} technique is not well suitable for our problem.   Other minor (non-fatal)  reasons    have already been mentioned in Section \ref{sec:1}. For example, the closed region with many holes shown in Figure \ref{fig:1new:c} illustrates the uncertainty region $u$, which is derived from Figure \ref{fig:1new:b} based on Equation (\ref{equ:temp compute u}). Clearly, to obtain $o.pcr(p)$ in Figure \ref{fig:1new:c} is more difficult than the case of no restricted areas (e.g., see Figure \ref{fig:1new:a}).

 To this step, it seems no better solution except the baseline method.  A casual trifle, shopping in a supermarket, gives us the initial inspiration (recall Section \ref{sec:1}). In the next section, we show  the details of our ideas, and then present the algorithm to answer the  explicit query.

 \section{Explicit CSPTRQ}\label{sec:eprq}

 \subsection{Spatial pruning/validating rules}\label{subsec:computation duality}
 For each object $o\in \mathscr O^*$, once we obtain the set $\mathscr R^*$ of candidate restricted areas,   the baseline method is to directly compute its uncertainty region $u$, and then to compute the  intersection result between $R$ and $u$. Let $s$ be the intersection result between $R$ and $u$,  it can be  formalized as follows.
 \begin{equation}
 s=u\cap R=(o.\odot -\bigcup _{r\in \mathscr R^*}r)\cap R
 \end{equation}

 Our method is to  \textit{swap} the order of geometric operations. The rationale behind it is surprisingly simple. Specifically, we  first compute ``$o.\odot \cap R$'', and then use the result of  ``$o.\odot \cap R$'' to subtract $\bigcup _{r\in \mathscr R^*}r$. It is formalized as follows.
 \begin{equation}\label{equation:swapping compuation}
 s=(o.\odot \cap R)-\bigcup _{r\in \mathscr R^*}r 
 \end{equation}
 
 There are  two significant  benefits by swapping the order of geometric operations. 
 
 (1) We can   prune some   objects, without the need of computing their uncertainty regions.
 Assume   ``q1'' shown in Figure \ref{fig:computation duality:a} is the query range $R$. Clearly, $o$  is a candidate moving object since $o.\odot_b$ intersects with $R_b$.  Here $o$  can be safely pruned  without the need of computing its uncertainty region $u$, since ``$R\cap o.\odot=\emptyset$''. Similarly, assume that ``q2'' is $R$. Here ``$R\cap o.\odot\neq \emptyset$''  (see Figure \ref{fig:computation duality:a}), but $(o.\odot \cap R)-\bigcup _{r\in \mathscr R^*} r=\emptyset$ (see Figure \ref{fig:computation duality:b}). Hence $o$   can also be pruned safely  without the need of computing $u$. 
 
 (2) We no longer need to  consider each $r\in \mathscr R^*$, which simplifies the computation of $s$.  For example, regarding to ``q2'', only the right most candidate restricted area is relevant with the computation of $s$. Similarly, regarding to ``q3'' shown in Figure \ref{fig:computation duality:b}, only two candidate restricted areas are relevant with the computation of $s$.

 \begin{figure}[t]
   \centering
   \subfigure[\small ]{\label{fig:computation duality:a}
      \includegraphics[scale=.43]{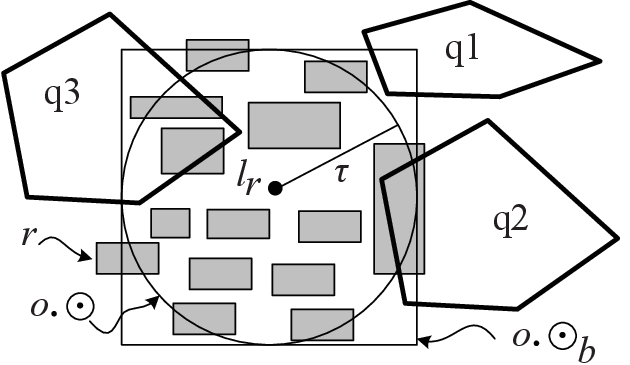}} 
      \hspace{2ex}
   \subfigure[\small ]{\label{fig:computation duality:b}
       \includegraphics[scale=.43]{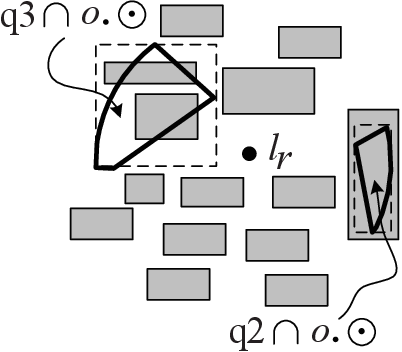}} 
             \vspace{-2ex}           
  \caption{\small   Example of swapping the order of gemetric operations} 
  \label{fig:computation duality}
 \end{figure}

 Hence, by swapping the order of geometric operations,  we can easily develop the following pruning/validating rules. 
 
 \begin{lemma}\label{Lemma:the first prune}
 Given the query range $R$ and an object $o\in \mathscr O^*$, we have
 \begin{itemize*}
 \item If $R\cap o.\odot=\emptyset$, then $o$ can be pruned safely.
 \item If $R\cap o.\odot=o.\odot$, then $o$ can be validated safely.
 \end{itemize*}
 \end{lemma}
 \noindent \textbf{Proof.}
 The proof is immediate by \textit{analytic geometry}. 
  {\raggedleft $\square$}
 
 Let ${\mathscr R}^\prime$ be a set of restricted areas such that the MBR of each $r\in {\mathscr R}^\prime$  has non-empty intersection set   with the MBR of $o.\odot \cap R$, we have an immediate corollary below. 
 
 \begin{corollary}\label{Lemma:the first prune copy}
 Given the query range $R$ and an object $o\in \mathscr O^*$, $o$ can be pruned safely if $(o.\odot \cap R)-\bigcup _{r\in \mathscr R^\prime}r =\emptyset$.  {\raggedleft $\square$}
 \end{corollary}

 \noindent \textbf{Discussion.}  
 We remark that the swapping operation itself is very easy, as it does not rely on any complicated technique. Furthermore, after we swap the order of geometric operations, to develop the pruning/validating rules is also not difficult.  
 We highlight it  because it is surprisingly simple but clearly efficient. 
 
 Now, for any object $o\in \mathscr O^*$, if it has not been pruned (or validated) by Lemma \ref{Lemma:the first prune} or Corollary \ref{Lemma:the first prune copy}, whether or not  we can directly compute its appearance probability $p$ using  Equation (\ref{equation:pj})? At the first sight,  it seems to be sure. 
 However, we should note that the intersection result $s$ obtained by Equation (\ref{equation:swapping compuation}) is possibly a fake result. We next share our insights and explain the details. 
 \subsubsection{Why is it possibly a fake result?}\label{subsec:why fake result}
 The  fake result  stems  from the \textit{location unreachability}. To explain it, we need  some basic concepts.

 Given $o.\odot$ and  a set $\mathscr R^*$ of candidate restricted areas, we say a  restricted area $r\in \mathscr R^*$ can subdivide $o.\odot$, \textit{if and only if} the result of ``$o.\odot-r$'' consists of multiple disjoint closed regions. We term each of those closed regions as  a \textit{subdivision}. 
 Let $\mathscr D$ denote the set of subdivisions, we say a subdivision $d\in \mathscr D$ is an  \textit{effective subdivision} such that $l_r\in d$, where $l_r$ is the (latest) recorded location of $o$ (recall Section \ref{subsec:problem setting}). 
 \begin{theorem}\label{theorem:unreachable}
 Assume that a  restricted area $r\in \mathscr R^*$ subdivides $o.\odot$, and $\mathscr D$ is the set  of subdivisions, if a subdivision $d\in \mathscr D$ is not the effective subdivision, then any point $p^\prime\in d$ is unreachable.  
 \end{theorem}
 \noindent \textbf{Proof.}
 It is easy to know that the object $o$ is located in $o.\odot$, as we adopt the \textit{distance based update policy}, recall Section \ref{subsec:problem setting}. We prove $p^\prime \in d$ is unreachable by contraction. Assume that $o$ can reach the point $p^\prime$, implying that there exists at least a path from $l_r$ to $p^\prime$ such that it does not directly pass through any restricted area and also the boundary of $o.\odot$. However, by the condition ``$d$ is not the effective subdivision'', implying that $l_r$ and $p^\prime$ are located respectively in two disjoint closed regions.  Based on \textit{analytic geometry}, it is clear that no such a path exists. This completes the proof.
  {\raggedleft $\square$}
 
 Theorem \ref{theorem:unreachable} gives us the insight into the location unreachability. See  Figure  \ref{fig:unreachability:c} for example, here the subdivision above $r_1$ and the one below $r_2$  are \textit{unreachable}.    Hence, $o$'s real uncertainty region, $u$, is the subdivision below $r_1$ and above $r_2$. 
 With this (concept) in mind, we next use a more  targeted example to show why $s$ obtained by Equation (\ref{equation:swapping compuation})  possibly is  a fake result. The shadow  region shown in  Figure \ref{fig:unreachability:a} or \ref{fig:unreachability:b}  illustrates $s$ obtained by Equation (\ref{equation:swapping compuation}), which is not equal to $\emptyset$. Here $o$  \textit{cannot} be pruned/validated based on Lemma \ref{Lemma:the first prune} and Corollary \ref{Lemma:the first prune copy}. The closed region with many holes shown in Figure \ref{fig:unreachability:b} illustrates $u$. For simplicity, assume that the location of $o$ follows uniform distribution in $u$. In this example, if we simply use the area of the shadow region to divide the area of $u$, we will get that $p$ is a positive number rather than 0. Clearly, it is a false answer, since $u$  and $s$  are disjoint, see Figure \ref{fig:unreachability:b}.
 
 \begin{figure}[t]
   \centering
   \subfigure[\small ]{\label{fig:unreachability:c}
      \includegraphics[scale=.6]{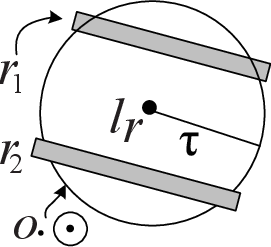}} 
      \hspace{4ex}
   \subfigure[\small ]{\label{fig:newfig:a}
      \includegraphics[scale=.6]{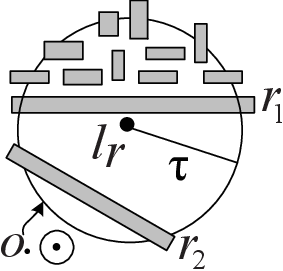}} 
             \vspace{-2ex} 
  \caption{\small Illustration of the location unreachability} 
  \label{fig:newfigure}
 \end{figure}

 

 \subsubsection{Natural solution}\label{subsubsub:why is fake result}
 
 To eliminate the fault produced by the above problem, the natural solution is   to compute $u$,  and then to check if $u$ intersects with   $s$.  If they are disjoint, then $p=0$ and $o$ should be pruned.  This approach  can indeed be used to eliminate the fault but it is   inefficient. 
 We next review two approaches \cite{ZhijieWang:prqumo} that are used to compute $u$, and then show the underlying reason.


 Given a closed region $c$,  we  let  $v^-$, $v^+$, $h^-$, $h^+$ denote the four (left, right, bottom, top) bounding lines of $c$, respectively. The \textit{span} of $c$ is $\mathop{\mathrm{arg max} }  \{dist(v^-,v^+),$ $dist(h^-,h^+) \}$, where $dist(\cdot)$ denotes the Euclidean distance.






 \begin{heuristic}
 Given  $o.\odot$, and two different candidate restricted areas, the candidate restricted area with the larger span is more likely to subdivide $o.\odot$ into multiple subdivisions. 
 \end{heuristic}
 
 To compute $u$, there are two approaches. The first one is using $o.\odot$ to subtract each restricted area $r\in \mathscr R^*$ one by one, and finally it chooses the subdivision \textit{containing} the point $l_r$ as the uncertainty region $u$ (see, e.g., Figure \ref{fig:unreachability:c}). 
 For ease of describing the second approach, we let $d^e$ denote the \textit{effective subdivision} (recall Section \ref{subsec:why fake result}), and \textit{slightly abuse}  the notation $d^e$.  
 
 
 The second one incorporates Heuristic 1, and can be generally described as follows. First, it  sorts  the set $\mathscr R^*$ of candidate restricted areas according to their \textit{spans} in the descending order  (implying that the restricted area $r\in \mathscr R^*$ with the larger span is to be handled firstly); and then it   uses  $o.\odot$ to subtract each $r\in \mathscr R^*$ one by one;  particularly,   when multiple subdivisions appear, it immediately chooses the effective subdivision $d^e$, and then uses  $d^e$ to subtract the next $r\in \mathscr R^*$, and so on; it finally gets  $u$ after all the restricted areas $r\in \mathscr R^*$ are handled.  See Figure \ref{fig:newfig:a},  $r_1$ is to be handled at first.   The subdivision below $r_1$ is taken as  $d^e$. Then, it uses $d^e$ to subtract $r_2$. Here,  the subdivision below $r_1$ and above $r_2$ is taken as  $d^e$. After this, the rest of restricted areas can be quickly pruned and thus do not need to execute (costly) geometric subtraction operations, improving the first approach. 
 
 \noindent \textbf{Why  is it inefficient?}  
 Consider  the example in Figure \ref{fig:unreachability:a} again, we can easily see that, if we want to get the uncertainty region $u$,  both of the approaches mentioned above  need to execute  subtraction operations many times. This justifies  the natural solution mentioned in the beginning of  Section \ref{subsubsub:why is fake result} is inefficient. Our strategy is to \textit{fight  poison with poison}.  In other words, we take advantage of the location uncertainty. This method  is pretty simple, but clearly efficient. The challenge is to find the \textit{point of penetration}, namely, when, where, and how to  take advantage of  the location unreachability.

 \begin{figure}[t]
   \centering
   \subfigure[\small ]{\label{fig:unreachability:a}
      \includegraphics[scale=.45]{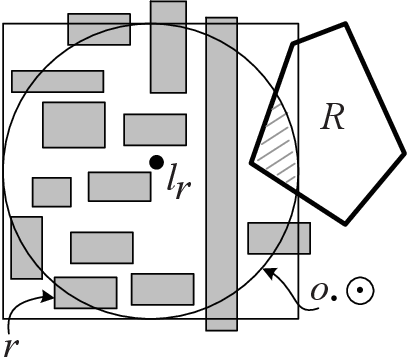}} 
      \hspace{2ex}
   \subfigure[\small ]{\label{fig:unreachability:b}
       \includegraphics[scale=.45]{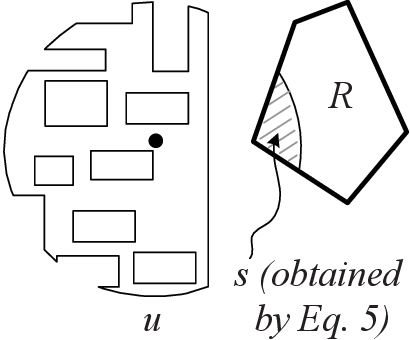}} 
             \vspace{-2ex} 
  \caption{\small Illustration of the fake result} 
  \label{fig:unreachability}
 \end{figure}

 \subsubsection{Take advantage of the location unreachability}\label{subsec:location unreacha}
 Based on the definition of \textit{subdivision}, the nature of \textit{location unreachability}, and  Equation (\ref{equation:swapping compuation}), we can  build the following theorem. 
 
 \begin{theorem}\label{theorem:always correct}
 Given $o.\odot$ and  $\mathscr R^*$,  $s$  (obtained by Equation (\ref{equation:swapping compuation}))    is  always  a correct result such that for any $r\in \mathscr R^*$, $|o.\odot-r|=1$, where  $|\cdot|$ denotes the number of subdivisions. 
 \end{theorem}
 
 Theorem 1 implies that the presence of multiple subdivisions (i.e., $|o.\odot-r|>1$) is an important sign of the  fault to be happened. Hence, if we \textit{correctly} and \textit{timely} handle this special case, the possible fault could be eliminated efficiently. With this (concept) in mind, a more efficient solution  comes into being.  Specifically, we  \textit{manage} to compute its uncertainty region $u$;  in the process of computing $u$, once {multiple} subdivisions appear, we also choose the effective subdivision $d^e$, but we do not directly use $d^e$ to subtract the next candidate restricted area.  Instead, we here  \textit{check} the geometric relation between $d^e$ and $s$ (obtained by Equation (\ref{equation:swapping compuation})).

 \begin{lemma}\label{lemma:prune 2}
 If $s\cap  d^e=\emptyset$, then $o$ can be pruned safely. 
 \end{lemma}
 \noindent \textbf{Proof.}
 The proof is not difficult but (somewhat) long, we move it to Appendix A. 
  {\raggedleft $\square$}

 We note that $s$ obtained by Equation (\ref{equation:swapping compuation})  possibly consists of multiple subdivisions. From Lemma \ref{lemma:prune 2}, we have an immediate corollary below. 
 \begin{corollary}\label{corollary:prune unrelated subdivisions}
 Given  $o.\odot$ and $R$, we assume $s$ (obtained by Equation (\ref{equation:swapping compuation})) consists of multiple subdivisions, say $s[1]$, $s[2]$, $\cdots$, $s{[|s|]}$, where $|s|$ is the total number of subdivisions in $s$. Without loss of generality, assume that $r\in \mathscr R^*$ can subdivide $o.\odot$ into multiple subdivisions, and $d^e$ is the effective subdivision. We have that, any subdivision $s[i]$ ($i\in [1,\cdots, |s|]$)  can be pruned safely if $s[i]\cap  d^e=\emptyset$.  
 $\square$
 \end{corollary}

 While this method is pretty simple, we can easily see that it gains  two benefits: it not only eliminates  the possible fault produced by Equation (\ref{equation:swapping compuation}), but also prunes  some objects in the early stages, without the need of  obtaining  the final results of their uncertainty regions. See, e.g., Figure \ref{fig:unreachability:a}, $o$ can be pruned after executing (only) one subtraction operation. 


 \noindent \textbf{Discussion.}
 All mechanisms discussed before belong to  \textit{spatial} pruning/validating mechanisms. For any object $o$  that has not been pruned/validated by the above  mechanisms, the natural method is  to compute its appearance probability $p$ using  Equation (\ref{equation:pj}) or (\ref{equation:area divide area}), and then to see if $p\geq p_t$, where $p_t$ is the so-called probabilistic threshold. In the next subsection, we present a more efficient method, which is initially inspired by \textit{evolutionary algorithms} (recall Section \ref{sec:1}). We remark that   $s$ (discussed in the rest of this paper) refers to  the correct result since we already eliminated  the possible fault.

 \subsection{Threshold pruning/validating rules}\label{subsec:multiple-step tactic}
 Our   method  computes  $p$ in a multi-step rather than   one-time way. We call it the multi-step mechanism.  Briefly speaking,  we first obtain a coarse-version result (CVR), which  is possibly  far away from the accurate value of  $p$. We then make a comparison between the CVR and  $p_t$, and check  if  $o$ can be pruned based on the current information. If otherwise, we  refine the CVR by the further computation. 
 
 
 \subsubsection{{Uniform distribution PDF}}\label{subsubsec:uniform pdf multiple step}
 To apply the multi-step mechanism to the case  of uniform distribution PDF,  we need to find appropriate \textit{carriers} (or  things) to which we can apply the multi-step mechanism.

 Suppose that there is a closed region with many holes. Its exact area clearly equals that the area of the closed region subtracts the areas of all holes. In contrast, if we  compute the area of the closed region, but do not subtract the areas of holes, we shall get  the most coarse result. Furthermore, we can easily see that this coarse result can be gradually refined by subtracting the rest of holes one by one. Hence,  the holes here are taken as the \textit{carriers}.  Based on this intuition, it is not difficult  to develop the followings. 
 
 For ease of understanding the details, we first should note that the uncertainty region $u$ is a single subdivision  (possibly) with holes; and  $s$ may be multiple   subdivisions (i.e., $|s|>1$) and each subdivision (possibly) has holes. Given a closed region $c$ with a  hole $h$, we say  the boundary of $c$ is the \textit{outer ring} of $c$,   and say the boundary of $h$ is the  \textit{inner ring} of  $c$. We also  use $\alpha (\cdot)$ to denote the area of a geometry. 
 
 
 Let $u_o$ be the outer  ring of uncertainty region $u$,   $u_h^i$ be the $i$th hole in $u$, and $\mathscr H$ be the set of holes in $u$, where $|\mathscr H|\geq 0$. We have 
 \begin{equation}\label{equation:computing UR}
 \alpha(u)=\alpha(u_o)-\sum_{i=0}^{|\mathscr H|} \alpha(u_h^i)
 \end{equation} 
 Similarly,  let $s[i]$ be the $i$th subdivision of $s$,   $s[i]_o$ be the outer ring of $s[i]$,  $s[i]_h^j$ be the $j$th hole in $s[i]$, and $|s[i]_h|$ be the number of holes in $s[i]$. We have
 \begin{equation}\label{equation:computing IS}
 \alpha(s)= \sum_{i=1}^{|s|} \alpha(s[i]) =\sum_{i=1}^{|s|} {\left( \alpha(s[i]_o)-\sum_{j=0}^{|s[i]_h|} \alpha(s[i]_h^j)\right)} 
 \end{equation} 
 For ease of presentation, we let $\mathscr H^*$ denote the set of (all) holes in $s$ (note: $|\mathscr H^*|=\sum_{i=0}^{|s|}|s[i]_h|$), and renumber these holes. Specifically, we let $s_h^j$  denote the $j$th hole among all the $|\mathscr H^*|$ holes. Therefore, Equation (\ref{equation:computing IS}) can be rewritten  as follows.
 \begin{equation}\label{equation:rewrite computing IS}
 \alpha(s) =\sum_{i=1}^{|s|} { \alpha(s[i]_o)} -\sum_{j=0}^{|\mathscr H^*|} \alpha(s_h^j) 
 \end{equation} 
 
 
 The natural solution (one-time way)  is to compute $\alpha(u)$ and $\alpha(s)$ based on Equation (\ref{equation:computing UR}) and (\ref{equation:rewrite computing IS}), respectively, and then to check  if  $\frac{\alpha(s)}{\alpha(u)}\geq p_t$. 
 
 In the proposed method, we also compute  $\alpha(u)$.  We however, do not directly compute  $\alpha(s)$. Specifically, we initially compute  $\sum_{i=1}^{|s|} \alpha(s[i]_o)$. Then, we compute the first CVR, denoted by $p^0$, as follows.
 \begin{equation}\label{equation:the first version}
 p^0=\frac{\sum_{i=1}^{|s|} \alpha(s[i]_o)}{\alpha(u)}
 \end{equation}
 \begin{lemma}\label{lemma:first version prune}
 Given  $p^0$ and the probability threshold $p_t$,  $o$ can be pruned safely if  $p^0$$<p_t$.
 \end{lemma}
 \noindent \textbf{Proof.}
 We only need to show that the appearance probability  $p$ is less than $p_t$. 
 Let $\epsilon$ denote an arbitrary non-negative number. We have 
 \begin{equation}\label{formula:non-negative number}
 p^0=\frac{\sum_{i=1}^{|s|} \alpha(s[i]_o)}{\alpha(u)}\geq \frac{(\sum_{i=1}^{|s|} \alpha(s[i]_o)-\epsilon}{\alpha(u)}
 \end{equation}
 In addition, since $p=\frac{\alpha(s)}{\alpha(u)}$,   by  Equation (\ref{equation:rewrite computing IS}), we have 
 \begin{equation}\label{equation:p j}
 p=\frac{\sum_{i=1}^{|s|} { \alpha(s[i]_o)} -\sum_{j=0}^{|\mathscr H^*|} \alpha(s_h^j)}{\alpha(u)}
 \end{equation}
 Clearly,  ``$\sum_{j=0}^{|\mathscr H^*|} \alpha(s_h^j)$'' in Equation (\ref{equation:p j}) is a non-negative number.  By Formula  (\ref{formula:non-negative number}) and Equation (\ref{equation:p j}), we have $p\leq p^0$. 
 Combining the condition ``$p^0<p_t$'', hence $p<p_t$. 
  {\raggedleft $\square$}
 
 If the object $o$ cannot be pruned based on Lemma \ref{lemma:first version prune}, and there exist holes in $s$, we further compute the second CVR,  and so on.  Let $p^{k-1}$ be the $k$th CVR, where $1< k\leq |\mathscr H^*|+1$. We have 
 \begin{equation}
 \label{equation:the second version prune}
 p^{k-1}=\frac{\sum_{i=1}^{|s|} \alpha(s[i]_o)-\sum _{j=0}^{k-1}\alpha(s_h^j) }{\alpha(u)}
 \end{equation}
 We should note that $p^{k-1}=p$ when $k=|\mathscr H^*|+1$. In other words, the final CVR is equal to the appearance probability   $p$. Furthermore,  $\sum _{j=0}^{k-1}\alpha(s_h^j)\leq\sum _{j=0}^{|\mathscr H^*|}\alpha(s_h^j)$, since $1< k\leq |\mathscr H^*|+1$. Hence, from Lemma \ref{lemma:first version prune}, we have an immediate corollary below.
 \begin{corollary}\label{corollary:the second prune}
 Given the $k$th CVR $p^{k-1}$ and the probability threshold $p_t$,   $o$  can be pruned safely if  $p^{k-1}<p_t$.  {\raggedleft $\square$}
 \end{corollary}

 \noindent \textbf{Discussion.} 
 We have shown how to apply the multi-step mechanism to the case of uniform distribution PDF, and developed new pruning rules. The  small challenge is to find appropriate \textit{carriers} to which we can apply the multi-step mechanism. To apply this mechanism to the case of   non-uniform distribution PDF, there is also a small challenge, which however, is different from the previous, as we can easily find appropriate \textit{carriers} by the similar observation. To explain this small challenge, we need some preliminaries. In the next subsection, we first introduce the preliminaries,  then clarify this small challenge, and finally  give the details of our method . 
 \subsubsection{{Non-uniform distribution PDF}}\label{subsubsec:nonuniform pdf multiple step}
 Regarding to the non-uniform distribution PDF,  a classical numerical integration method  is the Monte Carlo  method \cite{jinchuanchen:efficient,yufeitao:indexing,reynoldcheng:querying}. Let $N_1$ denote a pre-set value, where $N_1$ is an integer. 
 The natural solution is to randomly generate  $N_1$ {points} in the uncertainty region $u$. For each generated point $p^\prime$, it computes the value  $f(x_i,y_i)$  based on  its PDF, where $(x_i,y_i)$ are the coordinates of the point $p^\prime$, and then to check if $p^\prime \in s$. 
 Without loss of generality,   assume that  $N_2$  points (among  $N_1$ points)  are to be located in $s$. Then 
 \begin{equation}
 \label{equation:naive method compute pj}
 p=\frac{\sum _{i=1}^{N_2}f(x_i,y_i)}{\sum _{i=1}^{N_1}f(x_i,y_i)}
 \end{equation} 
 Finally, it checks  if $p$$\geq p_t$. If so, it puts the tuple $(o,p)$ into the result. Otherwise, $o$ is to be pruned.   
 
 We should note that  the Monte Carlo method is a non-deterministic algorithm. Thus we usually use a large sample as the input, in order to assure the accuracy of computation. Here the number of { generated points}  is  the size of sample. In general,  the larger  $N_1$ is, the \textit{workload error}  is more close to 0. Without loss of generality, assume that the  allowable workload error is $\delta$,  we can get the specific value of $N_1$ by the off-line test. 
 
 To this step,  we can easily realize that the generated points can be taken as \textit{carriers} to which we can apply the multi-step mechanism. In other words, the following steps are easily brought to mind: we initially generate a \textit{small} number of points, and thus get a coarse result; then, we  refine the previous coarse result by gradually  adding  points.  A small challenge is to construct the pruning rules. In other words,  assume that we get a coarse result,  how to decide whether or not $o$ can be pruned based on the current coarse result and the probabilistic threshold $p_t$. 
 
 To alleviate the small challenge above, we take advantage of the workload error.  Henceforth,  we  can easily  determine whether or not $o$ can be pruned based on three parameters: the current coarse result,  its corresponding workload error, and the probabilistic threshold $p_t$. We remark that the workload error  can also be  estimated  by the off-line test, when we use a \textit{small} number of  points. (In our experiments, we use the \textit{maximum}  workload error. For example, assume there are 100  approximate values, say $x_a^i$, where $i\in[1,100]$, and assume the exact value is $x_e$. Then,   the maximum workload error for this single value is $\mathop{\mathrm{arg max}}  \{ |x_e-x_a^i|\}$. By the extensive off-line test, an overall maximum workload error thus can be estimated. Again, the Monte Carlo method is a non-deterministic algorithm, thus the {extensive} off-line test is needed, in order to assure the accuracy of computation.)  
 Once we get rid of this small challenge, it is not difficult to develop the followings.

 Specifically, in the proposed method, we do not directly generate $N_1$  points; instead we initially generate $\lfloor \frac{N_1}{\theta} \rfloor$   points in $u$, where $\theta$ is an integer (e.g., 10). Let $N_2^0$ be the number of points   being located in $s$, where $N_2^0\leq \lfloor \frac{N_1}{\theta} \rfloor$.    Then, we get the first CVR  $p^0$ as follows.
 \begin{equation}\label{equation:the first version non uniform}
 p^0=\frac{\sum _{i=1}^{N_2^0}f(x_i,y_i)}{\sum _{i=1}^{\lfloor \frac{N_1}{\theta} \rfloor}f(x_i,y_i)}
 \end{equation}

 Let $\delta ^0$ be the workload error when we use $\lfloor \frac{N_1}{\theta} \rfloor$  points as the input. We have
 
 \begin{lemma}\label{lemma:non uniform first version prune}
 If $p^0$ $+\delta ^0<$$p_t$, then $o$ can be pruned safely.
 \end{lemma}
 \noindent \textbf{Proof.}
 Let $V_{\infty}$ be the  value obtained by Equation (\ref{equation:naive method compute pj})  when we set $N_1 \rightarrow + \infty$ (note: in this case the workload error can  be taken as  0).   It is clearly that $p^0-\delta ^0\leq V_{\infty}\leq p^0+\delta ^0$. Incorporating the condition ``$p^0$ $+\delta ^0<$$p_t$'', hence $V_{\infty}<p_t$. This completes the proof.
  {\raggedleft $\square$}
 
 If $o$ cannot be pruned based on Lemma \ref{lemma:non uniform first version prune}, we refine the first CVR by adding  points. For  the $k$th coarse-version, we denote by $\lfloor \frac{k\cdot N_1  }{\theta} \rfloor$, $\delta ^{k-1}$, and $N_2^{k-1}$ the  number of generated points, the workload error and the number of  points being located in $s$, respectively. Then, the $k$th CVR $p^{k-1}$ ($1< k\leq \theta$) can be derived as follows.
 \begin{equation}
 \label{equation:the second version prune non uniform}
 p^{k-1}=\frac{\sum _{i=1}^{N_2^{k-1}}f(x_i,y_i)}{\sum _{i=1}^{\lfloor \frac{k\cdot N_1}{\theta} \rfloor}f(x_i,y_i)}
 \end{equation}
 Furthermore, since each coarse-version corresponds to a workload error, from Lemma \ref{lemma:non uniform first version prune}, we have an immediate corollary below.
 \begin{corollary}\label{corollary:the second prune non uniform}
 Given the probability threshold $p_t$, the $k$th CVR $p^{k-1}$ and its corresponding workload error $\delta ^{k-1}$. If $p^{k-1}$ $+\delta ^{k-1}<$$p_t$, then $o$ can be pruned safely. {\raggedleft $\square$}
 \end{corollary}
 
 Up to now, we have shown all our pruning/validating rules (including spatial and threshold ones), we next pull them together to answer the explicit query.
 
 \subsection{Query processing for explicit CSPTRQ}\label{subsec:query processing for EPRQ}
 
 \subsubsection{{Algorithm}}
 
 

 Let  $\Re$ be the query result. Recall that ${\mathscr R}^\prime$ be a set of restricted areas such that the MBR of each $r\in {\mathscr R}^\prime$  has non-empty intersection set   with the MBR of $o.\odot \cap R$ (cf. Section \ref{subsec:computation duality}). Furthermore, we use $u[temp]$  to denote the  intermediate result of the uncertainty region $u$ (since we \textit{manage} to compute the uncertainty region $u$, and hope some objects can be pruned in the early stages, recall Section \ref{subsec:location unreacha}); similarly, we use $p[temp]$ to denote the intermediate result of $p$ (since we adopt multi-step way to compute the appearance probability $p$, recall Section \ref{subsec:multiple-step tactic}).    
 
 We first search the set $\mathscr O^*$ of candidate moving objects on the index  ${\mathscr I}_{o}$ using  $R_b$ as the input.  We then process each object $o\in \mathscr O^*$ based on  Algorithm 1 below. Note that in the following algorithm, the clause ``Discard $o$'' denotes that  the object $o$ is to be pruned, and we shift to process the next object, without the need of executing the remaining  lines.
 

 {\vspace{1ex}
 \small  \hrule
 \vspace{0.5ex}
 
 \noindent \textbf{Algorithm 1} \textit{Explicit CSPTRQ} 
 \vspace{0.5ex}
 
 \hrule
 \vspace{0.5ex}
 }

     {\footnotesize

 \noindent(1)$~~~$\textbf{if} $o.\odot \subseteq R$ 
 
 \noindent(2)$~~~$$~~~$Set $p\leftarrow 1$, and let $\Re\leftarrow\Re \cup $$(o,p)$ // $o$ be validated, Lemma \ref{Lemma:the first prune}

 \noindent(3)$~~~$\textbf{else if} $o.\odot \cap$$R=\emptyset$
 
 \noindent(4)$~~~$$~~~$Discard $o$ // $o$ be pruned,  Lemma \ref{Lemma:the first prune} 
 
 \noindent(5)$~~~$\textbf{else} // $o.\odot \cap$ $R\neq \emptyset$
 
 \noindent(6)$~~~$$~~~$Obtain ${\mathscr R}^\prime$  by searching on ${\mathscr I}_{r}$, and set $s\leftarrow$ ($o.\odot \cap R$)$-\cup _{r\in \mathscr R^\prime} r$

 \noindent(7)$~~~$$~~~$\textbf{if} $s=\emptyset$
 
 \noindent(8)$~~~$$~~~$$~~~~$Discard $o$ //  $o$ be pruned, Corollary \ref{Lemma:the first prune copy}
 
 \noindent(9)$~~~$$~~~$\textbf{else} // $s\neq \emptyset$ 
 
 \noindent(10)$~~~$$~~~$$~~$Obtain $\mathscr R^*$  by searching on ${\mathscr I}_{r}$,  set $u[temp]\leftarrow o.\odot$, and 
 
 $~~~$$~~~$$~~~~$sort all the restricted area $r$$\in \mathscr R^*$ according to their spans
 
 \noindent(11)$~~$$~~~$$~~~$\textbf{for} each $r$ $\in \mathscr R^*$
 
 \noindent(12)$~~$$~~~$$~~~$$~~~$Let $u[temp]\leftarrow u[temp]-r$
 
 \noindent(13)$~~$$~~~$$~~~$$~~~$\textbf{if} $|u[temp]|>1$ // multiple subdivisions appear
 
 \noindent(14)$~~$$~~~$$~~~$$~~~$$~~~$$u[temp]\leftarrow$Choose the effective subdivision

 \noindent(15)$~~$$~~~$$~~~$$~~~$\textbf{if} $u[temp]$ and $s$ are disjoint 
 
 \noindent(16)$~~$$~~~$$~~~$$~~~$$~~~$Discard $o$ //  $o$ be pruned, Lemma \ref{lemma:prune 2}

 \noindent(17)$~~$$~~~$$~~~$$~~~$\textbf{for} each $s[i]$ // $s[i]$ is a subdivision of $s$
 
 \noindent(18)$~~$$~~~$$~~~$$~~~$$~~~$\textbf{if} $u[temp]\cap s[i]=\emptyset$

 \noindent(19)$~~$$~~~$$~~~$$~~~$$~~~$$~~~$Remove $s[i]$ from $s$ // Corollary \ref{corollary:prune unrelated subdivisions}
 
 \noindent(20)$~~$$~~~$$~~~$Set $u\leftarrow u[temp]$
 
 \noindent(21)$~~$$~~~$$~~~$$p[temp]\leftarrow$ Compute the first CVR // Eq. \ref{equation:the first version} (or  \ref{equation:the first version non uniform})
 
 \noindent(22)$~~$$~~~$$~~~$\textbf{if} $p[temp]<$$p_t$ (or  $p[temp]+\delta ^0<p_t$)
 
 \noindent(23)$~~$$~~~$$~~~$$~~~$Discard $o$ // $o$ be pruned, Lemma \ref{lemma:first version prune} (or \ref{lemma:non uniform first version prune})
 
 \noindent(24)$~~$$~~~$$~~~$\textbf{else} 
 
 \noindent(25)$~~$$~~~$$~~~$$~~~$\textbf{while} $p[temp]$ is not the final CVR
 
 \noindent(26)$~~$$~~~$$~~~$$~~~$$~~~$$p[temp]\leftarrow$Compute the next CVR //Eq. \ref{equation:the second version prune} (or  \ref{equation:the second version prune non uniform})
 
 \noindent(27)$~~$$~~~$$~~~$$~~~$$~~~$\textbf{if} $p[temp]<p_t$ (or $p[temp]+\delta ^{k-1}<p_t$)
 
 \noindent(28)$~~$$~~~$$~~~$$~~~$$~~~$$~~~$Discard $o$ // $o$ be pruned, Corollary  \ref{corollary:the second prune} (or \ref{corollary:the second prune non uniform})
 
 \noindent(29)$~~$$~~~$$~~~$$~~~$Set $p\leftarrow p[temp]$, and let $\Re \leftarrow \Re \cup (o,p)$ // $o$ cannot be pruned by all the rules

 \hrule
 \vspace{1ex}
         }

 \vspace{1ex}
 



 We remark that we overlook the cost such as adding a tuple $(o, p)$ into $\Re$, comparing the geometric relation between two entities, etc., as these costs are trivial. Moreover, the span is a real number, hence the overhead to sort $|\mathscr R^*|$ candidate restricted areas is pretty small and  can (almost) be overlooked  compared to the overhead to execute $O(|\mathscr R^*|)$ times geometric subtraction operations. 
 In the sequel, we show how to extend techniques proposed in this section to answer the implicit query.
 \section{Implicit CSPTRQ}\label{sec:iprq}
 We first introduce the \textit{enhanced} multi-step computation, and then integrate the techniques proposed in  Section \ref{subsec:computation duality} to answer the implicit query. The enhance multi-step computation is easily brought to mind, as we have discussed the multi-step computation in the previous section, and we can easily see that the implicit query does not need to return the appearance probabilities of qualified objects, implying that some threshold validating rules can be developed. Note that the performance differences between the explicit and implicit queries stem mainly from this step.

 \subsection{Enhanced multi-step computation}\label{subsec:tactic1}
 
 The enhance multi-step strategy includes (\romannumeral 1) an \textit{adaptive} pruning/validating mechanism, which is used for the uniform distribution case,  and (\romannumeral 2) a \textit{two-way test} mechanism,  which is used  for the non-uniform distribution case. Regarding to the two-way test mechanism, there is no much surprise. Regarding to the first mechanism, its central idea is to cleverly choose appropriate rule (or method) according to the specific case. A small challenge  can be generally described as follows: given two methods and a specific case,  how to determine which method is more suitable for this specific case? In the sequel, we discuss more details. (Remark: most of notations discussed later actually have already been defined in previous sections, if any question, please refer to Table \ref{tab:main symbols} and/or Section \ref{subsec:multiple-step tactic}.)
 
 \subsubsection{{Adaptive pruning/validating mechanism}}\label{subsubsec:adaptive pruning and validate}
 Recall the tactic discussed in Section \ref{subsubsec:uniform pdf multiple step}. For the first coarse-version result (CVR), it is to compute  $\alpha(u)$ and  $\sum_{i=1}^{|s|} \alpha(s[i]_o)$ at first, and then to compute the first CVR $p^0$ based on  Equation (\ref{equation:the first version}). Since the implicit query does not need to explicitly return the probabilities of the qualified objects, clearly, it is also feasible that we  first  compute  $\alpha(s)$ and  $\alpha(u_o)$, and then  compute the first CVR $p^0$ as follows.
 \begin{equation}\label{equation:the frist version validate iprq}
 p^0=\frac{\alpha(s)}{\alpha(u_o)}
 \end{equation}
 \begin{lemma}\label{lemma:first version validate iprq}
 Given the probability threshold $p_t$ and the first CVR $p^0$ (obtained by Equation (\ref{equation:the frist version validate iprq})), we have that   if the first CVR $p^0$$>p_t$, then  $o$  can be validated safely.
 \end{lemma}
 \noindent \textbf{Proof.}
 We only need to show  $p$$>p_t$. The proof is the similar as the one of Lemma \ref{lemma:first version prune}. 
  {\raggedleft $\square$}
 
 If $o$ cannot be validated based on Lemma \ref{lemma:first version validate iprq}, and the number of holes in $u$ is not equal to 0 (i.e., $|\mathscr H|\neq 0$), we further compute the second CVR,  and so on. Then, the $k$th CVR $p^{k-1}$ ($1< k\leq |\mathscr H|+1$) can be derived as follows. 
 \begin{equation}
 \label{equation:the second version validate iprq}
 p^{k-1}=\frac{\alpha(s)}{\alpha(u_o)-\sum _{i=0}^{k-1}\alpha(u_h^i) }
 \end{equation}
 Note that $p^{k-1}$ equals the appearance probability $p$ when $k=|\mathscr H|+1$. Furthermore,  $\sum _{i=0}^{k-1}\alpha(u_h^i)$ $\leq$ $\sum _{i=0}^{|\mathscr H|}\alpha(u_h^i)$, since $1< k\leq |\mathscr H|+1$. Hence, from Lemma \ref{lemma:first version validate iprq}, we have an immediate corollary below.
 \begin{corollary}\label{corollary:the second validate iprq}
 Given the probability threshold $p_t$ and the $k$th  CVR $p^{k-1}$ (obtained by Equation (\ref{equation:the second version validate iprq})),   $o$ can be validated safely, if  $p^{k-1}>p_t$. {\raggedleft $\square$}
 \end{corollary}
 
 Hence, we can easily see that, if an object $o$ cannot be pruned/validated based on the spatial information, then  there are two methods to handle it.  
 \begin{itemize*}
 \item Method 1:  We compute the CVRs according to Equation (\ref{equation:the first version}) or (\ref{equation:the second version prune}), and then check if $o$ can be \textit{pruned} based on Lemma \ref{lemma:first version prune} or Corollary \ref{corollary:the second prune}.
 \item Method 2:  We compute the CVRs according to Equation (\ref{equation:the frist version validate iprq}) or  (\ref{equation:the second version validate iprq}), and then check if $o$ can be \textit{validated} based on Lemma \ref{lemma:first version validate iprq} or Corollary \ref{corollary:the second validate iprq}. 
 \end{itemize*}
 
 The naive solution is always to use  one of the two methods to handle those candidate moving objects that cannot be pruned/validated by the spatial information.  Instead, we adopt an \textit{adaptive} pruning/validating mechanism. In brief, if $o$ is more likely to be pruned, we use the ``Method 1''; in contrast, if $o$ is more likely to be validated, we use the ``Method 2''. Note that, there is a question ``given an object $o$, how to know it is more likely to be pruned (or validated)?''
 
 Specifically, we   compute a \textit{reference value}, which is used to estimate the \textit{trend} of $o$ (being more likely to be pruned/validated). Let $\gamma$  denote the reference value, which is computed  as follows.
 \begin{equation}\label{equation:referernce value}
 \gamma=\frac{\sum_{i=1}^{|s|} \alpha(s[i]_o)}{\alpha(u_o)}
 \end{equation}
 
 
 \begin{heuristic}
 Given $\gamma$ and $p_t$, if $\gamma<p_t$, then $o$ is more likely to be pruned. Otherwise, $o$ is more likely to be validated.
 \end{heuristic}
 
   
 Algorithm 2 below shows the pseudo codes of the adaptive pruning/validating mechanism. Lines 2-10 focus on pruning objects, and Lines 12-20 focus on validating objects. (Note that the meanings of the notations used in this algorithm are the same as the ones in Algorithm 1.)  
 
 \vspace{1ex}
 {\vspace{1ex}
 \small  \hrule
 \vspace{0.5ex}

 \noindent \textbf{Algorithm 2} \textit{Adaptive pruning/validating mechanism} 
 \vspace{0.5ex}
 
 \hrule
 \vspace{0.5ex}
 }
 
     {\footnotesize

 \noindent(1)$~~~~$$\gamma \leftarrow$Compute the \textit{reference value} // Equation (\ref{equation:referernce value})
 
 \noindent(2)$~~~~$\textbf{if} $\gamma<p_t$
 
 \noindent(3)$~~~~$$~~~~$$p[temp]\leftarrow$Compute the first CVR // Equation (\ref{equation:the first version})
 
 \noindent(4)$~~~~$$~~~~$\textbf{if} $p[temp]<$$p_t$
 
 \noindent(5)$~~~~$$~~~~$$~~~~$Discard $o$ // $o$ be pruned, Lemma \ref{lemma:first version prune}
 
 \noindent(6)$~~~~$$~~~~$\textbf{else}
 
 \noindent(7)$~~~~$$~~~~$$~~~~$\textbf{while} $p[temp]$ is not the final CVR
 
 \noindent(8)$~~~~$$~~~~$$~~~~$$~~~~$$p[temp]\leftarrow$Compute the next CVR // Equation (\ref{equation:the second version prune}) 
 
 \noindent(9)$~~~~$$~~~~$$~~~~$$~~~~$\textbf{if} $p[temp]<p_t$
 
 \noindent(10)$~~$$~~~~$$~~~~$$~~~~$$~~~~$Discard $o$ // $o$ be pruned, Corollary  \ref{corollary:the second prune} 
 
 \noindent(11)$~~$$~~~~$$~~~~$Let $\Re \leftarrow \Re \cup o$ // $o$ is a qualified object

 \noindent(12)$~~$\textbf{else} // $\gamma\geq p_t$
 
 \noindent(13)$~~$$~~~~$$p[temp]\leftarrow$Compute the first CVR // Equation (\ref{equation:the frist version validate iprq})
 
 \noindent(14)$~~$$~~~~$\textbf{if} $p[temp]\geq p_t$
 
 \noindent(15)$~~$$~~~~$$~~~~$Let $\Re \leftarrow \Re  \cup o$ // $o$ be validated, Lemma \ref{lemma:first version validate iprq}
 
 \noindent(16)$~~$$~~~~$\textbf{else}
 
 \noindent(17)$~~$$~~~~$$~~~~$\textbf{while} $p[temp]$ is not the final CVR
 
 \noindent(18)$~~$$~~~~$$~~~~$$~~~~$$p[temp]\leftarrow$Compute the next CVR // Equation (\ref{equation:the second version validate iprq})
 
 \noindent(19)$~~$$~~~~$$~~~~$$~~~~$\textbf{if} $p[temp]\geq p_t$
 
 \noindent(20)$~~$$~~~~$$~~~~$$~~~~$$~~~~$Let $\Re \leftarrow \Re  \cup o$ // $o$ be validated, Corollary  \ref{corollary:the second validate iprq} 
 
 \noindent(21)$~~$$~~~~$$~~~~$Discard $o$ // $o$ is  an unqualified object
 
 \hrule
 \vspace{1ex}
         }
 
 \vspace{1ex}
 \subsubsection{{Two-way test mechanism}}
 The two-way test mechanism is a simple extension of the method in Section \ref{subsec:multiple-step tactic}. For the sake of completeness, we present it below. 
 
 Regarding to the  first CVR, we can  also  compute it  according to Equation (\ref{equation:the first version non uniform}). Then, we have
 
 \begin{lemma}\label{lemma:two way the first version}
 Given the probability threshold $p_t$, the first CVR $p^0$ and its corresponding workload error $\delta ^0$, we have
 \begin{itemize*}
 \item If ``$p^0$ $+\delta ^0<p_t$'', then $o$ can be pruned safely.
 \item If ``$p^0$ $-\delta ^0\geq p_t$'', then $o$ can be validated safely. 
 \end{itemize*}
 \end{lemma}
 \noindent \textbf{Proof.}
 It is immediate by extending the proof of   Lemma \ref{lemma:non uniform first version prune}.
  {\raggedleft $\square$}
 
 If $o$ can be neither  pruned nor validated based on Lemma \ref{lemma:two way the first version}, we further compute  the second CVR, and so on. For the $k$th CVR, we can also compute it according to Equation (\ref{equation:the second version prune non uniform}).  From Lemma \ref{lemma:two way the first version}, we have an immediate corollary below.
 \begin{corollary}\label{corollary:two way the second}
 Given the probability threshold $p_t$, the $k$th CVR $p^{k-1}$ and its corresponding workload error $\delta ^{k-1}$, we have 
 \begin{itemize*}
 \item If ``$p^{k-1}$ $+\delta ^{k-1}<p_t$'', then $o$ can be pruned safely.
 \item If ``$p^{k-1}-$$\delta ^{k-1}\geq p_t$'', then $o$ can be validated safely. {\raggedleft $\square$}
 \end{itemize*}
 
 \end{corollary}
 
 The pseudo codes of the two-way test mechanism are shown in  Algorithm 3. We remark that in  the two-way test mechanism,  if  $o$  cannot (still) be pruned/validated by the final CVR,  we take the object $o$ as a qualified object, since  the final CVR equals $p$, and $p\in[p-\delta,p+\delta]$, where $\delta$ is the allowable workload error.

 \vspace{-1ex}

 \subsection{Query processing for implicit CSPTRQ}\label{subsec:query processing for IPRQ}
 
 \noindent \textbf{Algorithm.}
 The spatial pruning/validating mechanisms  proposed in Section \ref{subsec:computation duality}  can be seamlessly incorporated for answering the implicit query, implying that the algorithm for the implicit query  is the similar as the one for the explicit query. Specifically, we need to replace  Line 2 and Lines 21-29 in Algorithm 1 with new pseudo codes. Clearly, Line 2  should be replaced by ``${\Re\leftarrow \Re \cup o}$'', and Lines 21-29 should be replaced by the pseudo codes of the \textit{enhanced multi-step computation},  i.e., Algorithms 2 and 3.

 \noindent \textbf{I/O and query cost.}
 The I/O cost  is the same as the one of Algorithm 1. The query cost can be estimated using the similar method presented in Section \ref{subsec:query processing for EPRQ}. Specifically, the $i$ in Equation (\ref{equation:query cost}) should be replaced with a more small value, since the enhanced multi-step mechanism not only prunes but also validates objects.

 {\vspace{3ex}
 \small \hrule
 \vspace{0.5ex}

 \noindent \textbf{Algorithm 3} \textit{Two-way test mechanism} 
 \vspace{0.5ex}
 
 \hrule
 \vspace{0.5ex}
 }

     {\footnotesize

 \noindent(1)$~~~~$$p[temp]\leftarrow$ compute the first CVR // Equation  (\ref{equation:the first version non uniform})
 
 \noindent(2)$~~~~$\textbf{if} $p[temp]+\delta ^{0}<p_t$ $\parallel$ $p[temp]-\delta ^{0}\geq p_t$
 
 \noindent(3)$~~~~$$~~~~$\textbf{if} $p[temp]+\delta ^{0}<p_t$
 
 \noindent(4)$~~~~$$~~~~$$~~~~$Discard $o$ // $o$ be pruned, Lemma   \ref{lemma:two way the first version}
 
 \noindent(5)$~~~~$$~~~~$\textbf{else} // $p[temp]-\delta ^{0}\geq p_t$
 
 \noindent(6)$~~~~$$~~~~$$~~~~$$\Re \leftarrow \Re  \cup o$ // $o$ be validated, Lemma   \ref{lemma:two way the first version}

 \noindent(7)$~~~~$\textbf{else} 
 
 \noindent(8)$~~~~$$~~~~$\textbf{while} $p[temp]$ is not the final CVR
 
 \noindent(9)$~~~~$$~~~~$$~~~~$$p[temp]\leftarrow$Compute the next CVR // Equation (\ref{equation:the second version prune non uniform})
 
 \noindent(10)$~~$$~~~~$$~~~~$\textbf{if} $p[temp]+\delta ^{k-1}<p_t$ $\parallel$ $p[temp]-\delta ^{k-1}\geq p_t$
 
 \noindent(11)$~~$$~~~~$$~~~~$$~~~~$\textbf{if} $p[temp]+\delta ^{k-1}<p_t$
 
 \noindent(12)$~~$$~~~~$$~~~~$$~~~~$$~~~~$Discard $o$ // $o$ be pruned, Corollary   \ref{corollary:two way the second}
 
 \noindent(13)$~~$$~~~~$$~~~~$$~~~~$\textbf{else} // $p[temp]-\delta ^{k-1}\geq p_t$
 
 \noindent(14)$~~$$~~~~$$~~~~$$~~~~$$~~~~$$\Re \leftarrow \Re  \cup o$ // $o$ be validated, Corollary   \ref{corollary:two way the second}
 
 \noindent(15)$~~$$~~~~$$\Re \leftarrow \Re \cup o$

 \hrule
 \vspace{1ex}
         }
 
 \vspace{1ex}

 \section{Further optimization} \label{sec:further optimziation}
 In previous sections,  for each candidate moving object $o\in \mathscr O^*$ we retrieve the set $\mathscr R^\prime$ of restricted areas   from the database and then compute $s$, if $o$  cannot be pruned/validated by Lemma \ref{Lemma:the first prune}. Particularly, we  further retrieve the set $\mathscr R^*$ of restricted areas   from the database and then compute $u$,  if $o$ cannot still be  pruned by Corollary \ref{Lemma:the first prune copy}   (c.f., Algorithm 1). Note that there is an overlap between $\mathscr R^\prime$ and $\mathscr R^*$ (as $\mathscr R^\prime\subseteq \mathscr R^*$). This implies that in this case we  retrieve two times for the $|\mathscr R^\prime|$ restricted areas, which incurs the extra I/O cost. 
 With the similar observation we can also realize that for different candidate moving objects, their candidate restricted areas  may have an overlap.  This implies that previous methods retrieve multiple times  for those ``overlapped'' restricted areas, which also incurs the  extra I/O cost.

 To overcome the above limitations, we develop a novel strategy. The rationale behind this strategy is to track restricted areas that have been retrieved, avoiding to retrieve redundant data from the database. 
 Generally speaking, for each restricted area that has been retrieved from the database, we use the $<$key,value$>$ pair to store the ID and geometric data of restricted area in memory.   For brevity, we denote by  $D_m$ the data structure used to manage the set of $<$key,value$>$ pairs{\small  \footnote{\small Note that in our implementation, we employ the \texttt{map} container of C++ STL (standard template library).   }}.   
 Furthermore, we build another  R-tree, which is used to index restricted areas that have been retrieved. Here the leaf node does not store the detailed geometric data of restricted area, instead it only stores the ID and MBR of restricted area.  We denote by   $\mathscr I_r^\prime$ this R-tree for clearness.  With the help of  $D_m$ and $\mathscr I_r^\prime$, we can easily track the restricted areas that have been retrieved.  
 Specifically, we  do as  follows:
 \begin{itemize*}
 \item If we need to obtain $\mathscr R^\prime$ (or $\mathscr R^*$), we do not directly search on $\mathscr I_r$ and fetch restricted area data from the database. Instead, we first search on $\mathscr I_r^\prime$, getting a set, say $S_1$, of  IDs of restricted areas (note: these restricted area data can be obtained by accessing $D_m$ which is  stored in memory);  we then  search on $\mathscr I_r$, getting another set, say $S_2$, of IDs of restricted areas. Let $S_3$ be the set of IDs of restricted areas such that $S_3=S_2-S_1$. We here only need to fetch restricted area data (from the database) whose IDs are in $S_3$. The $|S_3|$ restricted areas fetched from the database and the $|S_1|$ ones obtained from $D_m$  constitute $\mathscr R^\prime$ (or $\mathscr R^*$). 
 \item If  restricted areas are fetched from the database, we immediately index these restricted areas using  $\mathscr I_r^\prime$, and add corresponding $<$key,value$>$ pairs into $D_m$. That is to say, we  update $\mathscr I_r^\prime$ and $D_m$ immediately once we fetched restricted area data from the database.     
 \end{itemize*}

 
 With the above concepts in mind, we can easily develop the improved algorithm for the explicit query.    
 First, we  search the set $\mathscr O^*$ of candidate moving objects on the index  ${\mathscr I}_{o}$ using  $R_b$ as the input. We then initialize $D_m$ and $\mathscr I_r^\prime$. Next, we process each object $o\in \mathscr O^*$. The  steps of processing each object $o$ are the similar as the ones in Algorithm 1, except that we need to make minor modifications on Lines 6 and 10 (here we use the strategy proposed in this section). Note that, the improved algorithm for implicit query is available by similar modifications. The pseudo codes of these two improved algorithms are immediate, and thus are omitted for saving space.     In the next section, we test the effectiveness and efficiency of the proposed algorithms, using extensive experiments under various experimental settings.

 
 


 

 
 
 

 
 
 
 
 



 \section{Experimental evaluation} \label{sec:6}
 
 \subsection{Experimental setup} \label{subsec:settings}
 Our experiments are based  on both real  and synthetic data sets, and the size of 2D space is fixed to   10000$\times$10000.  Two real data sets  called CA and LB{\small \footnote{\small The CA  is available in site: \url{http://www.cs.utah.edu/~lifeifei/SpatialDataset.htm},  and the LB  is available in site: \url{http://www.rtreeportal.org/} }}, are deployed. The data sets are the similar as the ones in \cite{reynoldcheng:querying,yunjungao:Continuous,yufeitao:indexing}. 
 The CA contains  104770  2D points, the LB contains  53145  2D rectangles. We let the CA   denote the (latest) recorded locations of moving objects, and the LB  denote the restricted areas. (Remark: this paper is not interested in querying the trajectories, and thus  does not use the trajectory data sets.) All data sets are normalized in order to fit the 10000$\times$10000 2D space.  Synthetic data sets also consist of two types of data. We generate a set of polygons to denote the restricted areas, and  place  them   in  this  space uniformly.   We generate a set of points to denote the (latest)  recorded locations of moving objects, and let them randomly distributed in  this space (note:  there is a constraint that these points cannot be located in the interior of any restricted area).  Moreover,  we randomly generate different \textit{distant thresholds} (between 20 and 50) for different moving objects. For brevity, we use the  CL  and RU  to denote the real (\underline{C}alifornia points together with \underline{L}ong Beach rectangles) and synthetic (\underline{R}andom distributed points together with \underline{U}niform distributed polygons) data sets, respectively.

 \begin{table}[t]
 \begin{center}
 \vspace{0ex}
 \begin{tabular}{p{.17\textwidth} p{.28\textwidth} p{.38\textwidth}  } \hline 
 
 {{\scriptsize Parameter}}&{{\scriptsize Description}}& {{\scriptsize Value}} \\ \hline

 {\scriptsize $N$}  & {\scriptsize number of moving objects}			&{\scriptsize \textbf{[}$10k,20k,30k,40k,\textbf{50k}$\textbf{]}} \\
 
 {\scriptsize $M$}  &{\scriptsize number of restricted areas}	&{\scriptsize \textbf{[}$10k,20k,30k,40k,\textbf{50k}$\textbf{]}} \\ 
 
 {\scriptsize $\zeta$}  &{\scriptsize number of edges of each $r$}	&{\scriptsize \textbf{[}$\textbf{4},8,16,32,64$\textbf{]}} \\

 {\scriptsize $\psi$}  &{\scriptsize number of edges of $R$}	&{\scriptsize \textbf{[}$\textbf{4},8,16,32,64$\textbf{]}}\\

 {\scriptsize $\epsilon$}  &{\scriptsize size of $R$}	&{\scriptsize \textbf{[}$100,200,300,400,\textbf{500}$\textbf{]}} \\ 
 
 {\scriptsize $p_t$}  &{\scriptsize probabilistic threshold}	&{\scriptsize \textbf{[}$0.1,0.3,0.5,\textbf{0.7},0.9$\textbf{]}} \\

 {\scriptsize $\eta$}  & {\scriptsize shape of $R$}	&{\scriptsize \textbf{[} \textbf{Sq},Ta,Dm,Tz,Cc \textbf{]}}\\ 
 
 
 {\scriptsize $N_1$}  &{\scriptsize  number of pre-set  points}	&{\scriptsize \textbf{[} 700 \textbf{]}}\\
 
 {\scriptsize $\theta$}  &{\scriptsize number of versions}	&{\footnotesize \textbf{[} 7 \textbf{]}}\\
 \hline
 \end{tabular}
 \end{center}
 \vspace{-2ex}
 \caption{\small Parameters Used in Our Experiments}\label{tab:experiment_parameters}
 \end{table}

 The performance metrics  include the preprocessing time, update time, I/O time and query time. Specifically,  the query time is the sum of I/O and CPU time.  The update time is the sum of the time  for updating the database record (i.e., $l_r$) and the one for updating the index ${\mathscr I}_{o}$, when an object reports its new location to the database server (note: we here do not consider the network transfer time). In order to investigate the update time, we randomly update  100 location records, and run 10 times for each  test, and then compute the average value for estimating a single location update. 
 To estimate the average  I/O and query time of a single query,  we  randomly generate  50 query ranges, and run 10 times for each  query range, and then compute the average value.
 Also, we run 10 times and compute the average value for estimating the preprocessing time.

 \begin{table}[h]
 \begin{center}
 
 \vspace{0ex}
 
 \begin{tabular}{p{.13\textwidth}  p{.79\textwidth} }\hline 

 {{\scriptsize Shape}}& {{\scriptsize Value}} \\ \hline 
 
 {\scriptsize Ta}  			&{\scriptsize \textbf{[}$(x,y),(x+L,y),(x+L/2,y+L)$\textbf{]}} \\
 
 {\scriptsize Tz}  	&{\scriptsize \textbf{[}$(x,y),(x+L,y),(x+2L/3,y+L),(x+L/3,y+L)$\textbf{]}} \\
 
 {\scriptsize Dm}  &{\scriptsize \textbf{[}$(x+L/2,y),(x+2L/3,y+L/3),(x+L,y+L/2),(x+2L/3,y+2L/3),(x+L/2,y+L),(x+L/3,y+2L/3),(x,y+L/2),(x+L/3,y+L/3)$\textbf{]}} \\
 
 {\scriptsize Cc} 	&{\scriptsize \textbf{[}$(x+L/3,y),(x+2L/3,y),(x+2L/3,y+L/3),(x+L,y+L/3),(x+L,y+2L/3),(x+2L/3,y+2L/3),(x+2L/3,y+L),(x+L/3,y+L),(x+L/3,y+2L/3),(x,y+2L/3),(x,y+L/3),(x+L/3,y+L/3)$\textbf{]}} \\ \hline
 \end{tabular}
 \end{center}
 \vspace{-2ex}
 \caption{\small Use Cases of $\eta$ }\label{tab:use case of Rshape}
 \end{table}

 Our experiments are conducted on a computer with 2.16GHz dual core CPU and 1.86GB of memory. The page size is fixed to 4K. The maximum number of children nodes  in the R-tree  ${\mathscr I}_{o}$ (${\mathscr I}_{r}$) is fixed to 50. The (latest) recorded locations of moving objects and the restricted areas    are stored using the MYSQL Spatial Extensions{\small \footnote{\small More information can be obtained  in site: \url{http://dev.mysql.com/doc/refman/5.1/en/spatial-extensions.html} }}. (Henceforth, we call them location records and restricted area records, respectively.) Other parameters  are listed in Table \ref{tab:experiment_parameters},  in which the numbers in \textbf{bold} denote the default settings.    $N$, $M$ and $\zeta$ are  the  settings of  synthetic data sets. The default setting of each restricted area $r$ is a rectangle with $40\times 10$ size. Sq, Ta, Dm, Tz and Cc denote  \underline{s}\underline{q}uare, \underline{t}ri\underline{a}ngle,  \underline{d}ia\underline{m}ond, \underline{t}rape\underline{z}oid and \underline{c}ross\underline{c}riss, respectively.  The specific settings of these geometries  are listed in Table \ref{tab:use case of Rshape}. These geometries are all bounded by the $500\times 500$ rectangular box (i.e., MBR).   $L$ in Table \ref{tab:use case of Rshape} is 500, and $(x,y)$ are   the coordinates of left-bottom point of its MBR, which are generated randomly. We use two types of PDFs:   \underline{u}niform \underline{d}istribution and \underline{d}istorted \underline{G}aussian. We use the UD and DG to denote them, respectively. In our experiments, the standard deviation  is set to $\frac{\tau}{5}$ (note: $\tau$ is the distance threshold), and the mean $u_x$ and $u_y$ are set to the coordinates  of the recorded location $l_r$. Following the guidance of \cite{ZhijieWang:prqumo}, we choose 700 as  the number of pre-set points. In addition, we use 7 coarse versions for the multi-step computation,  corresponding workload errors (WEs) are listed in Table \ref{tab:workload error}, these data are obtained by the off-line test. All workload errors refer to the {absolute workload errors}. More specifically,   $CV_7$ is the average (absolute) workload error, other versions are the \textit{maximum (absolute)  workload errors}. We remark that although  $\theta=7$ is not mandatory,   a too  small value weakens the efficiency of the multi-step mechanism, and a too large value  incurs not only over-tedious tests, but also negligible pruning/validating power between two consecutive versions.

 \begin{table}[h]
 \begin{center}
 \vspace{0ex}
 \begin{tabular}{ p{.1\textwidth}p{.06\textwidth}p{.06\textwidth}p{.06\textwidth}p{.06\textwidth}p{.06\textwidth}p{.06\textwidth}p{.06\textwidth}     }\hline 
 {{\scriptsize Property}}& {{\scriptsize $CV_1$}}& {{\scriptsize $CV_2$}}& {{\scriptsize $CV_3$}}& {{\scriptsize $CV_4$}}& {{\scriptsize $CV_5$}}& {{\scriptsize $CV_6$}}& {{\scriptsize $CV_7$}} \\ \hline 
 
 {\scriptsize $\lfloor \frac{k\cdot N_1 }{\theta} \rfloor $} 	&{\scriptsize 100 } 		&{\scriptsize 200 }	&{\scriptsize 300 }	&{\scriptsize 400 }	&{\scriptsize 500 }	&{\scriptsize 600 }	&{\scriptsize 700 }\\
 
 {\scriptsize WE}  	&{\scriptsize 0.3607} &{\scriptsize 0.2499}&{\scriptsize 0.2131}&{\scriptsize 0.1921}&{\scriptsize 0.1504}&{\scriptsize 0.1067}&{\scriptsize 0.0095} \\ \hline
 \end{tabular}
 \end{center}
 \vspace{-2ex}
 \caption{\small Multiple Version Workload Errors}\label{tab:workload error}
 \end{table}

 \subsection{Performance study }\label{subsec:results for eprq}
 As this paper is the first attempt to the CSPTRQ, the competitors are unavailable. 
 We implemented  the baseline method{\small \footnote{\small Note that, the efficiency of the baseline method for the explicit and implicit queries  are   identical; for ease of presentation, we  here do not differentiate them.}} (Section \ref{sec:preliminaries}),   the proposed methods  for the explicit (Section \ref{sec:eprq}) and implicit (Section \ref{sec:iprq}) queries, respectively. For brevity, we use the B,  PE and PI to denote the \underline{b}aseline method,  the \underline{p}roposed method for the \underline{e}xplicit query, and the \underline{p}roposed method for the \underline{i}mplicit query, respectively.  Note that  we  present the results for the explicit and implicit queries in a mixed manner, in order to save space.    We first  investigate the impact of  parameters $\psi$, $p_t$ and $\eta$ on the performance based on both real and synthetic data sets, and then study the impact of  parameters $N$, $M$, $\epsilon$,  $\zeta$ on the performance  based on synthetic data sets. Finally, we investigate the effectiveness of the optimization strategy proposed in  Section \ref{sec:further optimziation}.

 \begin{figure}[b]
   \centering
   \subfigure[\scriptsize RU  ]{\label{fig:exp:4e}
       \includegraphics[width=30ex,height=18.2ex]{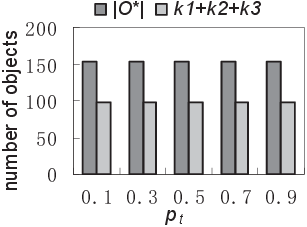}}
    \subfigure[\scriptsize CL ]{\label{fig:exp:4f}
        \includegraphics[width=30ex,height=18.2ex]{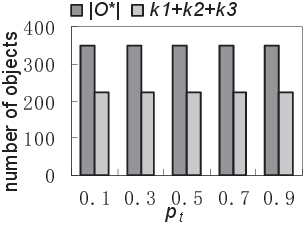}}
             \vspace{-2ex} 
  \caption{\small $|\mathscr O^*|$ and $k_1+k_2+k_3$ vs. $p_t$} 
  \label{fig:exp:4eandf}
 \end{figure}

 \begin{figure}[b]
  \centering
    \subfigure[\scriptsize ]{\label{fig:exp:1a}
        \includegraphics[width=11.4ex,height=18.2ex]{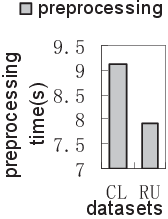}}
    \subfigure[\scriptsize ]{\label{fig:exp:1b}
        \includegraphics[width=10.1ex,height=18.2ex]{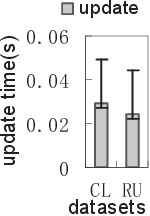}}
   \subfigure[\scriptsize ]{\label{fig:exp:1c}
       \includegraphics[width=32ex,height=18.2ex]{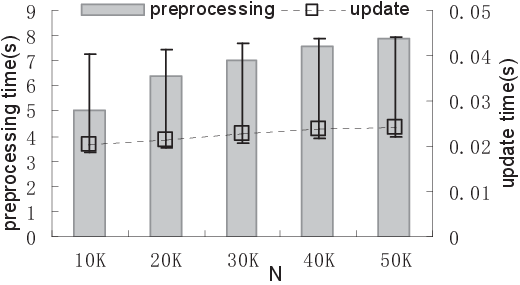}}\\
    \subfigure[\scriptsize ]{\label{fig:exp:1d}
        \includegraphics[width=32ex,height=18.2ex]{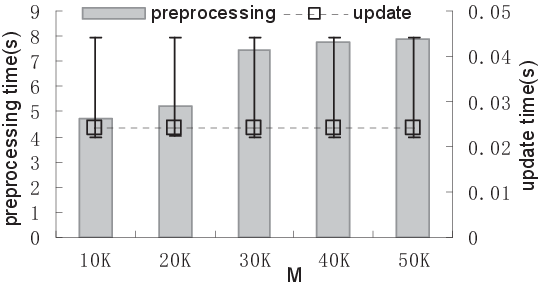}}
    \subfigure[\scriptsize ]{\label{fig:exp:1e}
        \includegraphics[width=32ex,height=18.2ex]{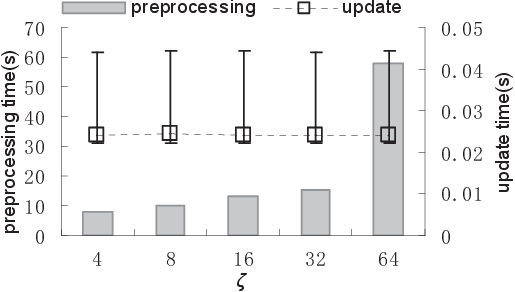}}
 
 \vspace{-2ex}
 \caption{\small Preprocessing and Update Performance}
 \end{figure}

 Thus far, all the experiments are based on both real and synthetic data sets. 
 For the two data sets, the  preprocessing time and update time are illustrated in Figure \ref{fig:exp:1a} and \ref{fig:exp:1b}, respectively. The preprocessing process is very fast, it only takes several seconds. (Note: recall Figure \ref{fig:1new:d}, the time is the \textit{hour} level if we pre-compute a set of uncertainty regions). Also, the update time is very short,  it only takes about tens of milliseconds.   In the sequel, we study the impact of  $N$,  $M$, $\epsilon$ and $\zeta$ on the performance, based on synthetic data sets.

 \begin{figure}[t]
   \centering
   \subfigure[\scriptsize RU  (UD)]{\label{fig:exp:6a}
       \includegraphics[width=31ex,height=19ex]{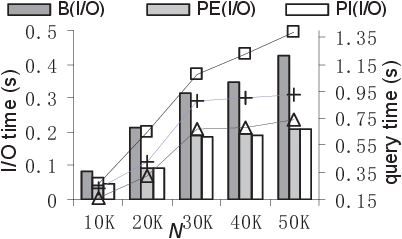}}
    \subfigure[\scriptsize RU  (DG)]{\label{fig:exp:6b}
        \includegraphics[width=31ex,height=19ex]{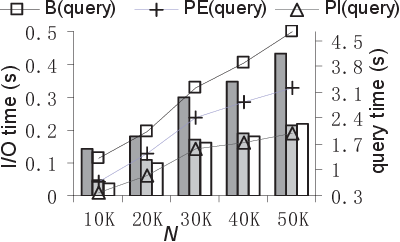}}
             \vspace{-2ex} 
  \caption{\small Query and I/O Efficiency vs. $N$} 
  \label{fig:exp:6}
 \end{figure}

 \noindent \textbf{Effect of $N$. }
 Figure \ref{fig:exp:1c} and Figure \ref{fig:exp:6}   illustrate the experimental results by varying  $N$ (the number of moving objects) from $1e+4$ to $5e+4$. From these figures, we can see that  the preprocessing time, update time, query time and I/O time increase as  $N$ increases.  In terms of the query  and I/O time,   the proposed methods always outperform the B, and   the (time) growth rate of the B is significantly faster than the ones of the proposed methods as  $N$ increases (especially when $N>3e+4$). This demonstrates that the proposed methods have better scalability. 
 
 \begin{figure}[b]
   \centering
   \subfigure[\scriptsize RU  (UD)]{\label{fig:exp:7a}
       \includegraphics[width=31ex,height=19ex]{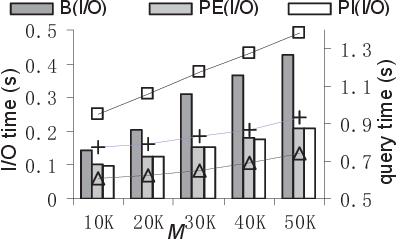}}
    \subfigure[\scriptsize RU  (DG)]{\label{fig:exp:7b}
        \includegraphics[width=31ex,height=19ex]{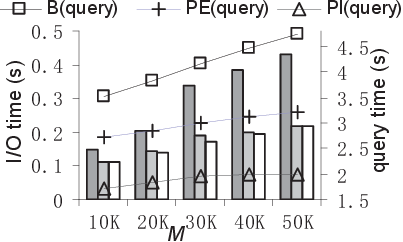}}
             \vspace{-2ex} 
  \caption{\small Query and I/O Efficiency vs. $M$} 
  \label{fig:exp:7}
 \end{figure}

 \noindent \textbf{Effect of $M$. }
 Figure \ref{fig:exp:1d} and  Figure \ref{fig:exp:7} illustrate the  results by varying  $M$ (the number of restricted areas) from $1e+4$ to $5e+4$. We can see from Figure \ref{fig:exp:1d} that  the preprocessing time increases as  $M$ increases, whereas the update time is constant as  $M$ increases.  This is because the preprocessing process needs to construct  ${\mathscr I}_{r}$ (the index of restricted areas); the update process however, is irrelevant with ${\mathscr I}_{r}$. In addition, Figure \ref{fig:exp:7} shows that both the query  and I/O time slightly increase as  $M$ increases, and the proposed methods always outperform the B. Similar to the last set of experiments, in terms of the query and I/O time,  the growth rate of the B is significantly faster than the proposed methods as  $M$ increases. This further demonstrates that the proposed methods have better scalability.

 Up to now, we have reported the main experimental results related to the baseline method and proposed methods. We are now ready to investigate the effectiveness of the optimization strategy proposed in Section \ref{sec:further optimziation}. With regard to explicit and implicit queries,  we  use respectively the PE$+$O and PI$+$O  to denote the  algorithms integrated the \underline{o}ptimization strategy presented in Section \ref{sec:further optimziation}, for ease of discussion. 
 
 \noindent \textbf{Effectiveness of optimization strategy. }
 Figure \ref{fig:exp:1star1} reports the results when explicit queries are executed. From this figure we can easily see that the I/O time of PE$+$O is obviously less than the one of PE, i.e., the improvement factor{\small \footnote{\small Here the improvement factor refers to the ratio of   time.  Assume that the I/O time of PE is 0.8736 seconds and the one of PI$+$O is 0.274 seconds  
 for example, the improvement factor is  $\frac{0.8736}{0.0.274}=3.189$.    }} is relatively large. This demonstrates that the strategy proposed in Section \ref{sec:further optimziation} is effective. Note that the query time of PE$+$O is also less than the one of PE (although the improvement factor is not as much as the one for I/O time).   Figure \ref{fig:exp:1star2} reports the results when implicit queries are executed, from which we can derive  similar findings. We remark that  when we vary other parameters (e.g., $\xi$, $N$, $M$) instead of $\zeta$, the experimental results also support our findings, i.e., the PE$+$O (PI$+$O) outperforms the    PE (PI), and the improvement factor for I/O time is relatively large.  To save space, we here do not plot those  results. 
  
  \begin{figure}[t]
    \centering
    \subfigure[\scriptsize RU dataset (UD)]{\label{fig:exp:1star1}
        \includegraphics[width=31ex,height=19ex]{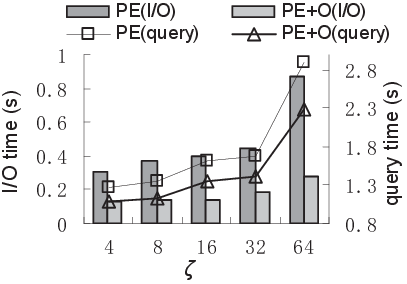}}
     \subfigure[\scriptsize RU dataset (UD)]{\label{fig:exp:1star2}
         \includegraphics[width=31ex,height=19ex]{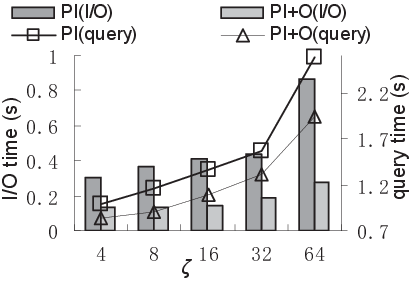}}
              \vspace{-2ex} 
   \caption{\small The effectiveness of optimization} 
   \label{fig:exp:1star1and2}
  \end{figure} 
  
 In addition to testing the total I/O time, we also investigate the I/O time for retrieving restricted areas and  moving objects, respectively. Figure \ref{fig:exp:2star1and2} reports the results when the default settings are used. We can easily see that in terms of PE,  most of I/O time are spent on  retrieving restricted area data  from the database. In contrast, the PE$+$O takes less time to retrieve restricted areas, as the optimization strategy discussed in Section \ref{sec:further optimziation}  avoids to retrieve redundant restricted area data from the database.    Another interesting finding is that when the CL data sets are used, the effectiveness of optimization strategy is more obvious. This is because the points (i.e., recorded locations of moving objects) are clustered in the CL data sets, rendering that different candidate moving objects easily share the same  restricted areas. We remark that the I/O time of implicit query is the same as the one of explicit query, omitted for saving space.

 \begin{figure}[h]
   \centering
   \subfigure[\scriptsize RU  ]{\label{fig:exp:2star1}
       \includegraphics[width=30ex,height=18.2ex]{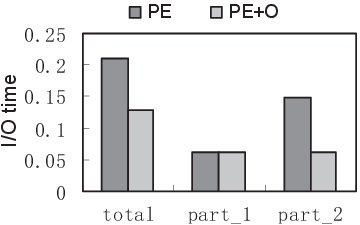}}
    \subfigure[\scriptsize CL ]{\label{fig:exp:2star2}
        \includegraphics[width=30ex,height=18.2ex]{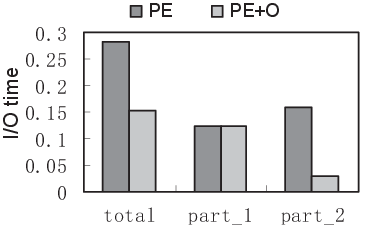}}
             \vspace{-2ex} 
  \caption{\small Total I/O time and patial I/O time. In these figures, the term ``part$\_$1''  denotes the I/O cost for  retrieving moving objects, and the term ``part$\_$2'' denotes the I/O cost for retrieving restricted areas.} 
  \label{fig:exp:2star1and2}
 \end{figure} 
 
 \noindent \textbf{Summary. }
 On the whole,  these experimental results show us that (\romannumeral 1) the proposed algorithms  obviously outperform the baseline method regardless of the I/O or query performance; (\romannumeral 2) the proposed algorithms have  better scalability, compared to the baseline method; (\romannumeral 3) while the I/O performance of two proposed algorithms is  identical,  they have different query performance (it is consistent with our theoretical analysis); (\romannumeral 4) the preprocessing process is fast and the update efficiency is high; (\romannumeral 5) the optimization strategy (discussed in Section \ref{sec:further optimziation}) can significantly improve the I/O efficiency, and also reduce the query time although the improvement factor is not very large.    
 Furthermore, these experimental results also demonstrate the  robustness and flexibility of our methods.

 
 



 \section{Concluding remarks} \label{sec:7}
 In this paper, we discussed the CSPTRQ for moving objects. We  differentiated two forms of CSPTRQs: explicit and implicit ones (as they can have  different  solutions, performance results, and purposes/applications). We showed the challenges, and proposed efficient solutions that are easy-to-understand and also easy-to-implement. Interestingly, the initial idea of our solutions is inspired by a casual trifle --- shopping in a supermarket. In brief, to answer the explicit query, we incorporated two main ideas: \underline{s}wapping the order of geometric operations;  and computing the probability using a \underline{m}ulti-step mechanism. We then extended these ideas to answer the implicit query, in which an \underline{e}nhanced multi-step mechanism is naturally developed. Furthermore,  we developed a novel strategy used to retrieving restricted areas in a more efficient manner. While the rationales behind our solutions are  simple, extensive experimental results  demonstrated  the effectiveness and efficiency of the proposed algorithms. Meanwhile, from the experiential results, we further perceived the difference between explicit and implicit queries; this interesting  finding is  meaningful for the future research. In the future, we prepare to study other  types of probabilistic threshold queries (e.g., \textit{concurrent queries}, \textit{kNN queries}) while considering the existence of restricted areas (i.e., obstacles).  Another interesting  research topic is to extend the concept of restricted areas  to  other uncertainty models.


{ \normalsize
\bibliographystyle{abbrv}
\bibliography{sample}
}


\end{document}